\documentclass[sigconf]{acmart}

\renewcommand\footnotetextcopyrightpermission[1]{}

\usepackage{amsmath}
\usepackage{amsfonts}
\usepackage{graphicx}
\usepackage{color}
\usepackage{url}
\usepackage{subfigure}
\usepackage{xspace}
\usepackage{tabularx}
\usepackage{balance}
\usepackage{multirow}
\usepackage{booktabs}
\usepackage[linesnumbered,algoruled,boxed,lined]{algorithm2e}
\usepackage{tabu}

\usepackage{cleveref}

\usepackage{algorithmic}

\newtheorem{definition}{Definition}
\newtheorem{theorem}{Theorem}

\newtheorem{proposition}{Proposition}

\renewcommand{\paragraph}[1]{\smallskip \noindent \textsf{\textbf{#1}}}

\DeclareMathOperator*{\argmin}{arg\,min}

\renewcommand{\qed}{\hfill $\framebox(6,6){}$}

\newcommand{\mech}{{M}\xspace}
\newcommand{\data}{\mathcal{D}\xspace}

\newcommand{\cthm}{Theorem~}

\newcommand{\cprop}{Proposition~}

\newcommand{\ie}{{\em i.e.}\xspace}
\newcommand{\eg}{{\em e.g.}\xspace}

\newcommand{\squishlist}{
	\begin{list}{$\bullet$}{
		\setlength{\itemsep}{0pt}
		\setlength{\parsep}{3pt}
		\setlength{\topsep}{3pt}
		\setlength{\partopsep}{0pt}
		\setlength{\leftmargin}{1.0em}
		\setlength{\labelwidth}{1em}
		\setlength{\labelsep}{0.5em}
   }
}

\newcommand{\squishenum}{
	
	\begin{list}{\usecounter{scount}}{
		\setlength{\itemsep}{0pt}
		\setlength{\parsep}{3pt}
		\setlength{\topsep}{3pt}
		\setlength{\partopsep}{0pt}
		\setlength{\leftmargin}{1.2em}
		\setlength{\labelwidth}{1em}
		\setlength{\labelsep}{0.5em}
	}
}

\newcommand{\squishend}{
	\end{list}
}

\newcommand{\ep}[2][]{{\bf E}_{#1}\hspace{-0.06cm}\left[#2\right]\xspace}

\newcommand{\bigoh}[1]{{\rm O}\!\left(#1\right)\xspace}
\newcommand{\bigohinline}[1]{{\rm O}(#1)\xspace}

\newcommand{\mdl}{{\rm MDL}\xspace}

\begin{document}

\title{A Pluggable Learned Index Method via \\ Sampling and Gap Insertion}

\author{Yaliang Li$^*$, Daoyuan Chen$^*$, Bolin Ding, Kai Zeng, and Jingren Zhou} 
\affiliation{Alibaba Group}
\affiliation{ \{yaliang.li, daoyuanchen.cdy, bolin.ding, zengkai.zk, jingren.zhou\}@alibaba-inc.com}
\thanks{$^*$The first two authors contributed equally to this work and are joint first authors.}

\begin{abstract}
Database indexes facilitate data retrieval and benefit broad applications in real-world systems. 
Recently, a new family of index, named learned index, is proposed to learn hidden yet useful data distribution and incorporate such information into the learning of indexes, which leads to promising performance improvements.
However, the ``learning'' process of learned indexes is still under-explored.
In this paper, we propose a formal machine learning based framework to quantify the \textit{index learning objective}, and study two general and pluggable techniques to enhance the \textit{learning efficiency} and \textit{learning effectiveness} for learned indexes. 
With the guidance of the formal learning objective, we can efficiently learn index by incorporating the proposed sampling technique, and learn precise index with enhanced generalization ability brought by the proposed result-driven gap insertion technique.

We conduct extensive experiments on real-world datasets and compare several indexing methods from the perspective of the index learning objective. The results show the ability of the proposed framework to help to design suitable indexes for different scenarios. 
Further, we demonstrate the effectiveness of the proposed sampling technique, which achieves up to 78x construction speedup while maintaining non-degraded indexing performance.
Finally, we show the gap insertion technique can enhance both the static and dynamic indexing performances of existing learned index methods with up to 1.59x query speedup.
We will release our codes and processed data for further study, which can enable more exploration of learned indexes from both the perspectives of machine learning and database.
\end{abstract}

\maketitle

\pagestyle{plain}

\section{Introduction}
\label{sec:intro}
Database indexes have been extensively studied in the database community for past decades, resulting in fruitful designed index methods and broad real-world applications, \eg, \cite{abel1984b+, bitmap1,Pagh2004cuckoo,wang2017survey}. The related topics become even more important in the era of big data as tremendous data are generating and collecting from numerous sources at every second. 
Recently, a new family of index, namely learned index, has attracted increasing attention due to its promising results in terms of both index size and query time \cite{kraska2018case,Ferragina:2020pgm,ding2019alex,galakatos2019fiting}.

The research direction of learned index is opened up by \cite{kraska2018case}, which regards the traditional indexes such as B+Tree as models that predict the actual location within a sorted array for a queried key. 
From this view, indexes can be machine learning models trained from data to be indexed, and the hidden distribution information of data (for example, patterns and regular trends) can be leveraged to optimize the ``indexing'' components of traditional indexes.
For example, to achieve small index size, the \textit{index storage layout} of B+Tree, i.e., the height-balanced tree, can be replaced by a small tree whose nodes are machine learning models with a few learnable parameters \cite{kraska2018case}, or a small B+Tree whose nodes maintain the learned parameters of a few piece-wise linear segments \cite{galakatos2019fiting} rather than the whole data.
To achieve fast query speed, the \textit {index querying algorithm} of B+Tree is changed from traversing internal nodes with multiple comparisons and branches to inferences of machine learning models with a few numeric calculations \cite{Ferragina:2020pgm, ding2019alex}.

However, the ``learning'' of learned indexes is still under-explored in existing learned index methods. First, how to compare different learned indexes with evaluation from the perspective of machine learning, such as learning objectives and capacity of chosen models?  For different data with varied hidden patterns and different resource constraints, a formal learning objective can help us to design and optimize learned index. 
Second, few existing learned index methods study the cost of index learning while the sizes of real-world data are heavily increasing. Can we reduce the learning cost of learned indexes to enhance their applicability for large-scale datasets? 
Third, existing learned index methods explore few about distributions of the data to be learned. Can we deeply explore the data distribution information to further improve the learned index performance? 

Motivated by these three questions, in this paper, we propose a formal machine learning based framework to measure the \textbf{index learning objective}, and study two general and pluggable learning-oriented techniques, \ie, \textit{sampling} to enhance the \textbf{learning efficiency}, and data redistribution via \textit{gap insertion} to improve the index performance from the view of \textbf{learning effectiveness}.

First, to formally quantify the learning objective and learning quality of learned indexes, we regard the learned indexes as encoding mechanisms that compress data information into learned models and measure the learned index using the minimum description length (MDL) principle \cite{grunwald2007minimum}. By formulating an MDL-based objective function as the learning objective, we connect two important concepts of machine learning, overfitting and underfitting, to the flexible learning of indexes. With the help of the index learning objectives, we discuss how to compare existing learned index methods and design suitable learned indexes for different scenarios.

Second, to reduce the learning cost and accelerate the index construction especially for large-scale datasets, we propose to learn index with sampled small subsets. By capturing the data distribution information hidden in the data, it is possible to learn the index with a small subset of data while achieving high learning effectiveness. We theoretically prove the feasibility of learning index with sampling, and provide the asymptotic guideline on how large the sample should be in order to learn an index having a comparable performance with the index learned on the whole dataset.

Last but not the least, to enhance the performance and generalization ability of off-the-shelf learned indexes with few extra efforts, we study what a data distribution can be beneficial to index learning, and propose a data re-distribution technique via gap insertion. 
Specifically, we design a result-driven strategy to estimate the gaps that should be inserted in a data-dependent manner, and propose a gap-based data layout and strategy to physically place keys on positions with gaps inserted. 
In comparison to the original distribution, the re-distributed data is easier to be learned and can boost the performance of static indexing operations.
Surprisingly, the result-driven gap insertion also enables us to handle dynamic indexing scenarios well, which is due to the reason that the estimated gap-inserted positions can be predictively reserved for possible inserted keys in the dynamic workloads.

We conduct comprehensive experiments on four wildly adopted real-world datasets. 
We compare several index methods \cite{abel1984b+,galakatos2019fiting,Ferragina:2020pgm,kraska2018case} from the perspective of the proposed MDL-based framework, and examine these methods in terms of learning objective, model regularization and model capacity, providing a new understanding of them.
Further, we apply the proposed two general learning-oriented techniques, \textit{sampling} and data re-distribution via \textit{gap insertion} into existing learned index methods. 
Promising improvements are achieved for both these two techniques: 
The proposed sampling technique achieves up to 78x construction speedup, meanwhile maintaining non-degraded query performance and reasonable prediction preciseness;
The gap insertion technique significantly improves the preciseness of learned indexes and achieves up to 1.59x query speedup over strong learned index baselines.
Finally, we also show that the learned indexes with gap insertion can achieve good performance on dynamic workloads. 

To summarize, we make the following contributions: 
(1) We propose an MDL-based framework that enables formal analysis of learning objectives and comparison of different learned indexes, and more importantly, the framework can guide us to design flexible learned indexes for various scenarios; 
(2) We propose a pluggable \textit{sampling} technique that can improve the learning efficiency of learned index, which is practical and useful for index construction acceleration, especially on large-scale datasets;
(3) We propose a pluggable technique, data re-distribution via \textit{gap insertion}, to better utilize the hidden distribution information of indexed data, which improves the preciseness and generalization ability of learned indexes;
(4) We conduct comprehensive experiments to verify the effectiveness of the proposed techniques, and we will release our codes and datasets to promote further studies.
With these contributions, we hope to better connect the database community with the machine learning community to enhance each other for the topic of learned indexes.

\section{Preliminary}
\label{sec:pre}

\paragraph{Indexes as Mechanisms}.
Let's first formalize the task of learning index from data: Given a database $\data$ with $n$ records (rows), assume that a {\em range index} structure will be built on a specific column $x$. For each record $i \in [n]$, the value of this column, $x_i$, is adopted as the {\em key}, and $y_i$ is the {\em position} where the record is stored (for the case of {\em primary indexes}), or the {\em position} of the pointer to the record (for the case of {\em secondary indexes}). $y_i$ can be interpreted as the position of $x_i$ in an array sorted by $x_i$, and to support range queries, a database to be indexed needs to satisfy the key-position monotonicity.

\begin{definition}[Key-position Monotonicity]
	\label{def:monotonicity}
	For a set of key-position pairs $\{(x_i,y_i)\}$, the key-position monotonicity means that, for any $x_i$ and $x_j$, $y_i < y_j$ iff $x_i < x_j$. 
\end{definition}

The task of learned index aims to train and learn a predictive {\em index mechanism} $M(y|x)$ from $\data$: $\mech$ takes a record's key $x$ as input and outputs a {\em predicated position} $\hat{y} \leftarrow \mech(x)$ in the sorted array for data access.
From this perspective, classic index structures such as B$+$ tree \cite{abel1984b+, jagadish2005idistance}, B$^*$ tree \cite{chang2000b} and Prefix B-tree \cite{bayer1977prefix} can be also regarded as such a mechanism designed by experts, which gives the exact position or the page number for a given key.

\paragraph{Prediction-Correction Decomposition}.
The learning of an index mechanism $M$ is essentially to approximate the joint distribution of keys and positions.
Ideally, the predicted position $\hat{y}$ should be exactly the same as the true position of a record $x$, i.e., $\hat{y} = y$. However, in general, $\hat{y}$ is different with $y$ since it is difficult to perfectly fit a real-world dataset containing complex patterns. 
Thus, the query process of $x$ can be decomposed into a \textit{``prediction''} step that gives prediction $\hat{y}$ based on $x$ using $M$, and a \textit{``correction''} step that finds the true position $y$ of indexed record based on $\hat{y}$.

For example, considering B-Tree index, it gives the page number $\hat{y}$ where the record $x$ is located in, and requires a further page scan to get the exact position of this particular record.
For learned indexes, after getting the predicted position $\hat y$ through machine learning model inference, we also need to conduct, \eg, a binary/exponential search around $\hat y$ to find the true position $y$ where the record is stored. The cost of this \textit{``correction''} step depends on the difference $|\hat{y} - y|$.
Indeed, one goal is to minimize the difference between $\hat{y}$ and $y$. 
But the difference can be non-zero, which provides the flexibility to adjust the cost of \textit{``prediction''} step for a mechanism, \ie, how much space we need to store $\mech$ and how much time we need to calculate $\mech(x)$. 
This raises several questions: can we evaluate a learned index from the costs of these two decomposed steps? And what roles do these two steps play in the index learning?

\section{An MDL-based Framework}
\label{sec:framework}
In this section, we formally quantify the index mechanisms using the Minimum Description Length (MDL) principle  \cite{grunwald2007minimum, grunwald2019minimum}.
We first connect the above prediction-correction decomposition of index mechanisms with MDL, then formulate the learning objective of index mechanisms with two general terms that can encode various indexing criteria, and finally discuss the physical meanings of the proposed concepts with several instantiations of index mechanisms.

\subsection{Measuring Mechanism with Minimum Description Length}
The idea of MDL is to regard both the given data $\data$ and a mechanism $M$ as codes, and regard the learning as appropriately encoding or compressing the data $\data$ using $M$.
Then we can use the \textit{code length} or \textit{description length} to measure the ``simplicity'' of a mechanism: The more we learn the hidden patterns in the data and reduce its redundancy, the shorter the description lengths for the compressed data and the learned mechanism.
Specifically, the description length of a mechanism $M$ is decomposed into $L(M)$ and $L (\data|M)$: $L(M)$ indicates the description length of the mechanism itself, and $L (\data|M)$ indicates the conditional description length of $\data$ given $M$, which can be interpreted as how many extra bits do we need to exactly describe $\data$ using $M$. 
In the context of learned index, we can link $L(M)$ and $L(\data|M)$ to the \textit{prediction} and \textit{correction} respectively. 

Formally, in this paper, we will use the MDL criteria to measure the quality of an index mechanism $\mech$ as  $\mdl(M, \data)$ $= L(\mech) + L (\data|\mech)$. To be more specific, the two terms:
\begin{itemize}
	\item $L(\mech)$ measures the cost of using $\mech$ for prediction, \ie, getting the predicted position $\hat{y}$ based on the key $x$, which is usually proportional to the size of $\mech$;
	\item $L(\data|\mech) = \mathbb{E}_{\{x,y\}\in \data}L(y, \hat{y})$ measures the average cost of getting the true position $y$ based on the predicted position $\hat{y}$, in the position correction step.
\end{itemize}

\paragraph{Two Example Instantiations.}
To better illustrate the idea of MDL, let's consider two instantiations of mechanism $\mech$, classic B-tree and polynomial function. 
For a classic B-tree, $L(\mech)$ denotes the cost of traversing index tree, which is usually proportional to the height of the index tree $h$, i.e., $L(\mech) \propto h$. The concrete form of $L(\mech) = f(h)$ depends on several implementation details of B-tree, such as the maximum number of children a node can have, or the minimum number of children an internal (non-root) node can have. The second term $L (\data|\mech)$ denotes the cost of page scan in the leaf node, which is a function of page size $s$, i.e., $L (\data|\mech) = f(s)$. When binary scan is adopted, $f(s)$ is a linear function of $\log(s)$. 

For the second example, let's consider the case that a polynomial function is learned to act as the mechanism, that is $\mech(x) = \sum_{t=0}^{T}a_t \cdot x^t + c_t$, where $a_t$ and $c_t$ are learnable parameters for a corresponding degree $t$, and $T$ is the highest degree. In this example, $L(\mech)$ can be measured as $\bigoh T$, because $\bigoh T$ multiplications are needed to get the predicted position $\hat{y}$ based on the key $x$. For the second term in MDL, $L(\data|\mech) = \mathbb{E}_{\{x,y\}\in \data}L(y, \hat{y}) = \mathbb{E}_{\{x,y\}\in \data} (\log(|y-\hat{y}|) + 1)$ if binary search is adopted for the correction step.

\subsection{Objective Function for Learned Indexes}
From the above discussion of two instantiations, we can see the MDL provides a criterion to formally analyze and compare different indexes, including both classic B-tree and machine learning based indexes. This enables us to design suitable index structures for various scenarios. 
Formally, given the data $\data$ and a family of possible mechanisms $\mathcal{\mech}$, the process of learning index from data can be formulated as to find the best mechanism $\mech^*$ that minimizes the description length as follows:
\begin{equation}
\mech^* = \argmin_{\mech \in \mathcal{\mech}} \mdl(M, \data) = \argmin_{\mech \in \mathcal{\mech}} \left( L(\mech) + \alpha L ({\mathcal D}|\mech) \right),
\label{eq:overall}
\end{equation}
where $\alpha$ is a coefficient to balance the two terms in MDL.
From the view of machine learning, $\data$ is fed into the learning procedure stated in Equation (\ref{eq:overall}) and the description length $(M,\data)$ acts as the objective function to be minimized. 
By fitting the dataset, the learning procedure finds an optimal mechanism $\mech$ to predict the position $y$ based on given $x$. 
If a mechanism achieves minimum description length, it learns and stores the underlying key-position distribution information about data in the most compact form.  
For example, given a toy data $\data = \{(2, 1), (4, 2), (5, 3), (6, 4), (8, 5)\}$, the learned index can represent it as $\mech(x) = \text{Round}(0.7x - 0.5)$. Instead of storing the key-position pairs, only parameters of $\mech$ are needed to store, which reduces the redundancy of raw data and compresses the information to store.
Although the distributions of real-world data are not so simple, various patterns can be mined by machine learning methods, and be expressed in compact forms.

In Equation (\ref{eq:overall}), three factors $L(M)$, $L(D|M)$ and $\alpha$ act as the knobs to tune the performance of learned indexes, and next we will give more discussion about these factors.

\paragraph{Choice of $L(M)$ and $L(D|M)$.}
To learn an index mechanism from data, the first step is to choose the concrete forms of these two terms in the above framework. As the list of keys and corresponding positions $\{(x, y)\}$ are given and the underlying physical storage format is fixed, we can first choose a specific search method for the correction step, and thus the concrete form of $L (\data|\mech) = \mathbb{E}_{\{x,y\}\in \data}L(y, \hat{y})$ can be determined.
For $L(\mech)$, it has more flexibility. It can be set as the number of operations to calculate $\mech(x)$,  the number of model parameters in $\mech$, or the summarization of the $p$-norm of all model parameters, etc. The choice of $L(\mech)$ can be made by considering the requirements of database systems and the constraints of computation resources (\eg, on-disk or in-memory) to fit various scenarios. 
For example, existing learned index method PGM \cite{Ferragina:2020pgm} implicitly adopts $L (\data|\mech)=log(|\hat{y}-y)|)+1$ since it uses a binary search, and $L (\mech)$ to be the number of learned model parameters as it uses an optimal linear segmentation learning algorithm.

\paragraph{Overfitting v.s. Underfitting.}
When learning to build the index from data, the coefficient $\alpha$  in Equation (\ref{eq:overall}) plays an important role. From the perspective of machine learning,  $ L(D|M) = \mathop{\mathbb{E}}_{\{x,y\} \in D} L(y, \hat{y})$ measures the prediction loss on training data $D$, while $ L(M)$ is  the regularization term of learned model. These two terms are usually contradicted: We can learn a very complex model M to make the prediction loss on training data zero or close to zero, which leads to a small $ L(D|M)$ while a large $L(M)$, the so-called overfitting phenomenon \cite{bishop2006pattern,shalev2014understanding}; On the other hand, if we learn a simple model, $L(M)$ is small while  $L(D|M)$ is large as the model is too simple to make precious predictions, the so-called underfitting phenomenon. 
The coefficient $\alpha$ makes a trade-off between these two terms and aims to learn an index mechanism having the minimum description length, i.e., a relatively simple model $M$ while also having a small prediction loss on training data. 

\paragraph{$\alpha$ in Existing Index Methods.}
\label{sec:alpha-in-MDL}
For existing index methods, there are some tunable parameters playing the role of $\alpha$ to adjust the $L(M)$ term and $L(D|M)$ term in MDL.
For example, the page size of B+ Tree, the tree depth and the layer width of RMI \cite{kraska2018case}, and the error bound $\epsilon$ of FITting-Tree \cite{galakatos2019fiting} and PGM \cite{Ferragina:2020pgm}. These parameters are to be tuned for a given $\data$ and reflect the degree we want to fit $\data$:
With a small page size, a deep and wide model tree, and a small $\epsilon$, these index mechanisms will substantially gain small perdition errors $L(\data|M)$ and usually large $L(M)$ such as large index sizes or long prediction times.
In other words, these parameters implicitly act as the regularization factors in index learning.
In the experiments (Section \ref{exp:trade-off-MDL}), we will investigate the trade-offs of several index mechanisms with different $\alpha$s in more detail.

\section{Learning Index With Sampling}
\label{sec:sample}
As formulated above, we aim to choose $\mech$ from a candidate family ${\mathcal M}$ to minimize the objective function $L(\mech) + \alpha L({\mathcal D}|\mech)$, for a given dataset ${\mathcal D}$. 
When $|\mathcal D|$ is large, it is an expensive learning task. 
In fact, it is expensive just to evaluate the loss $L({\mathcal D}|\mech)$. 
In this section, we will investigate a computationally efficient solution via sampling.

\subsection{Accelerating Index Construction}
Recall that the cost of the correction step (corresponding to the term $L({\mathcal D}|\mech)$ in the MDL objective function) is proportional to $\log|y - M(x)| = \log|y - \hat{y}|$ when binary search is adopted (refer to above \emph{Two Example Instantiations}). Let's focus on the objective function in the following form:
\[
L(\data|\mech) = \ep[(x,y) \in \data]{L(y, \mech(x))} = \ep[(x,y) \in \data]{\log|y-\mech(x)|}.
\]
Let $\mathcal{E}$ be the maximum absolute prediction error, i.e., the maximum absolute value of all differences between predicated positions and true positions: $\forall (x, y) \in \data, ~|M(x) - y| \leq \mathcal{E}$. 
We have such an upper bound $L(y, \mech(x)) $ $\leq \log \mathcal{E}$, for any $(x,y)\in \data$ and a binary or an exponential search. Thus, we can show that a small sample from $\mathcal D$ suffices to approximate $L(\data|\mech)$:
\begin{proposition}
	\label{prop:sample_loss}
	Given that $L(y, \mech(x)) \leq {\log \mathcal{E}}$, we draw a random sample $\data_s$ from $\mathcal D$ with $|\data_s| = n_s$. For any candidate mechanism $\mech$, we can estimate $L(\data|\mech)$ using 
	\[
	{L}(\data_s|\mech) = \frac{1}{n_s} \sum_{(x,y) \in \data_s} L(y, \mech(x)),
	\]
	s.t., with probability as most $1-\delta$, we have
	\[
	\big|{L}(\data_s|\mech) - L(\data|\mech)\big| \leq \frac{\log \mathcal{E}}{\sqrt{2n_s}} \cdot \sqrt{\log\frac{2}{\delta}}.
	\]
\end{proposition}
\begin{proof}
	Here we provide a proof sketch due to the space limitation. Let's consider the random sample $\{(x,y) \in \data_s\}$. Indeed, we have
	\[
	\ep{{L}(\data_s|\mech)} = \frac{1}{n_s} \sum_{(x,y) \in \data_s} \ep{L(y, \mech(x))} = L({\mathcal D}|\mech).
	\]
	Since $L(y, \mech(x)) \in [0, \log \mathcal{E}]$, we apply Hoeffding's inequality to finish the proof.
\end{proof}

We can interpret $L(\data|\mech)$ as the expected cost of the correction step. Theoretically, we only need to estimate it with an error up to a constant factor (\eg, no more than the size of a page). To this end, we only need to draw a small sample.

\begin{theorem}
	\label{thm:loss_approx}
	Consider the optimization problem: 
	\[
	\mech^* = \argmin_{\mech \in \mathcal{\mech}} \mdl(\mech, {\mathcal D}) = \argmin_{\mech \in \mathcal{\mech}} \left( L(\mech) + \alpha L ({\mathcal D}|\mech) \right).
	\]
	We can solve it on a random sample $D_s$ with size $s = \bigohinline{\alpha^2 \log^2 \mathcal{E}}$ as
	\[
	\hat\mech^* = \argmin_{\mech \in \mathcal{\mech}} \mdl(\mech, \data_s)
	\]
	s.t., $\mdl(\hat\mech^*, {\mathcal D}) \leq \mdl(\mech^*, {\mathcal D}) + \bigoh{1}$ with high probability.
\end{theorem}
\begin{proof}
	Indeed, we have
	\[
	\begin{aligned}
	& ~\mdl(\hat\mech^*, {\mathcal D}) - \mdl(\mech^*, {\mathcal D})
	\\
	= & ~\mdl(\hat\mech^*, {\mathcal D}) - \mdl(\hat\mech^*, S) + ~\mdl(\hat\mech^*, \data_s) 
	\\
	& - \mdl(\mech^*, \data_s) + ~\mdl(\mech^*, \data_s) - \mdl(\mech^*, {\mathcal D})
	\\
	\leq & ~|{L}(\data|\hat\mech^*) - L(\data_s|\hat\mech^*)| + 0 + |{L}(\data|\mech^*) - L(\data_s|\mech^*)|,
	\end{aligned}
	\]
	where $\mdl(\hat\mech^*, \data_s) - \mdl(\mech^*, \data_s) \leq 0$ if from the optimality of $\hat\mech^*$. Using \cprop\ref{prop:sample_loss}, both terms $|{L}(\data|\hat\mech^*) - L(\data_s|\hat\mech^*)|$ and $|{L}(\data|\mech^*) - L(\data_s|\mech^*)|$ can be bounded to complete the proof.
\end{proof}

We can interpret \cthm\ref{thm:loss_approx} as follows. For the purpose of minimizing the MDL function considered in Equation (\ref{eq:overall}), it suffices to learn the index $\mech$ on a small random sample with size as small as $\bigoh{\alpha^2\log^2 \mathcal{E}}$. 
Further, the parameter $\alpha$ controls how much we weigh the cost of the correction step in relative to the size of the index $L(\mech)$ in our goal. The larger the weight is, the larger sample we need to draw in order to learn the index; in other words, more details about the data are needed to learn an index in finer granularity. It should be noticed that there must be some constants hidden in the big O notation, which does matter in practice; however, the above theorem enables us to speed up index learning with small samples, and serves as an asymptotic guideline on how large the sample needs to be.  We will show that the proposed sampling technique can achieve significant construction speedup while maintaining comparable query performance in experiments (Section \ref{exp:sample}).

\subsection{Discussion}
\label{sec:sample-disscuss}
\paragraph{$\mathcal{E}$ in Existing Learned Index Methods}.
Now let's examine the maximum prediction error $\mathcal{E}$ for several existing learned index methods. 
FITing-Tree \cite{galakatos2019fiting} and PGM \cite{Ferragina:2020pgm} introduce an error bound $\epsilon$ to limit the maximum absolute difference between the actual and predicted position of any key, and thus $\mathcal{E}$ can be set as $\mathcal{E} = \epsilon$.
For RMI \cite{kraska2018case}, recall that the maximum prediction errors, $(y'-y)$ where $y'$ can be larger or less than $y$, are recorded during the training, and thus $\mathcal{E}$ can be bounded by the maximum value of them: $\mathcal{E}=max(|max\_positive\_error|,$ $|min\_negative\_error|)$. 
For these existing learned index methods, we can see that $\mathcal{E}$s are usually far smaller than $|\data|$, indicating that we can learn an index mechanism efficiently via sampling a small subset of $\data$.
In experiments (Section \ref{exp:sample}), we will see that these three learned index methods require different numbers of samples to achieve non-degraded performances and their numbers of samples meet our asymptotic analysis $\bigoh{\alpha^2\log^2 \mathcal{E}}$ with different $\alpha$s and different $\mathcal{E}$s.

\paragraph{Unseen Keys in Sampling}.
Interestingly enough, the sampling technique can bridge two indexing scenarios, static indexing and dynamic indexing, which correspond to two machine learning concepts, overfitting and generalization.
In static indexing scenarios, we only need to consider how to overfit $\data$ since all $\{x_i, y_i\}$ to be indexed can be accessed during the learning. 
However, when learning on a sampled subset data $\data_s$, the learned index mechanism $M$ has to consider the model generalization ability to precisely predict positions for \textbf{unseen data} $\data_s^-=\data - \data_s$. While for dynamic indexing scenarios including possible inserted data, similarly, a learned index mechanism $M$ needs to be generalized to unseen new keys.
This inspires us that if we can learn an indexing mechanism having good generalization ability from a small sampled data $\data_s$, we will be able to learn a dynamic index that has good generalization ability and can handle inserted keys.
Thus, as shown in the example of Figure \ref{fig:sample_demo}, besides index construction speedup, the sampling technique also improves the model generalization ability.
This seems contradicted with machine learning theory: usually the generation bound can be improved when we have more data, not fewer data. However, for learned index, the data pairs are not independent and a small sample can be enough to learn the correlation between keys and their positions. Moreover, the sampling technique can help to learn the underlying distribution of $\data$ better by extracting more general patterns (red segments in Figure \ref{fig:sample_demo}) from $\data_s$ and alleviate the overfitting (blue segments in Figure \ref{fig:sample_demo}) caused by some noisy keys in the unseen $\data_s^-$.

\begin{figure}[htbp]
	\centering
	\subfigure[Sampling]{
		\begin{minipage}{0.2\textwidth}
		\centering
		\includegraphics[width=\textwidth]{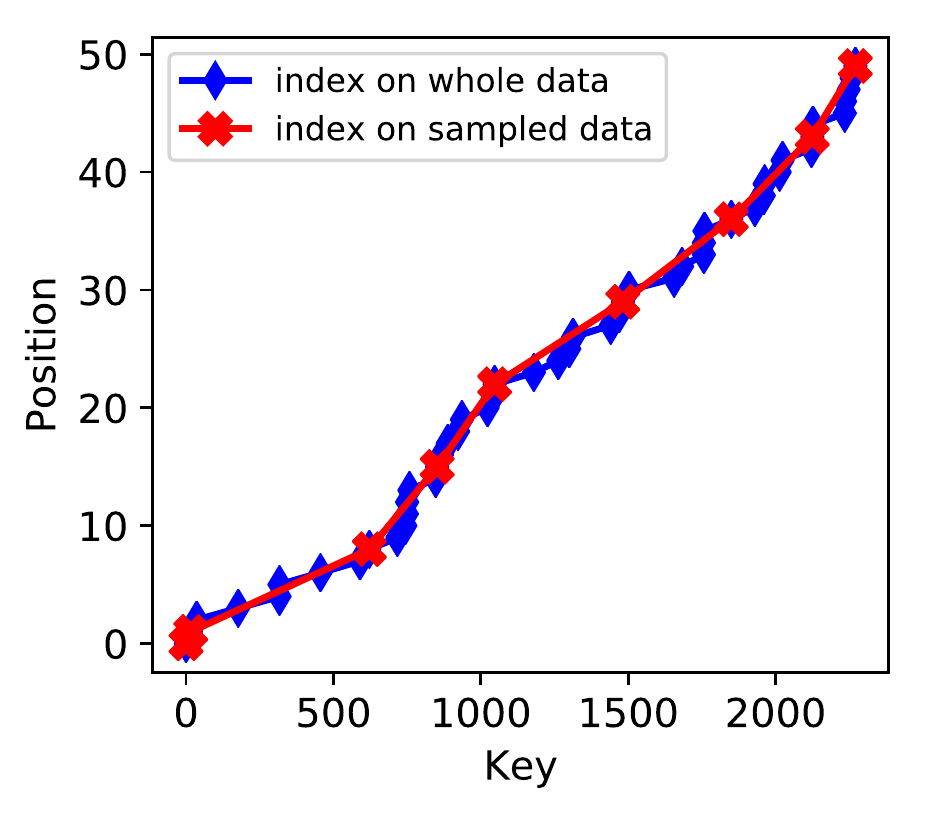}
		\label{fig:sample_demo}
		\vspace{-0.12in}
		\end{minipage}
	}
	\subfigure[Gap Insertion]{
		\begin{minipage}{0.2\textwidth}
		\centering
		\includegraphics[width=\textwidth]{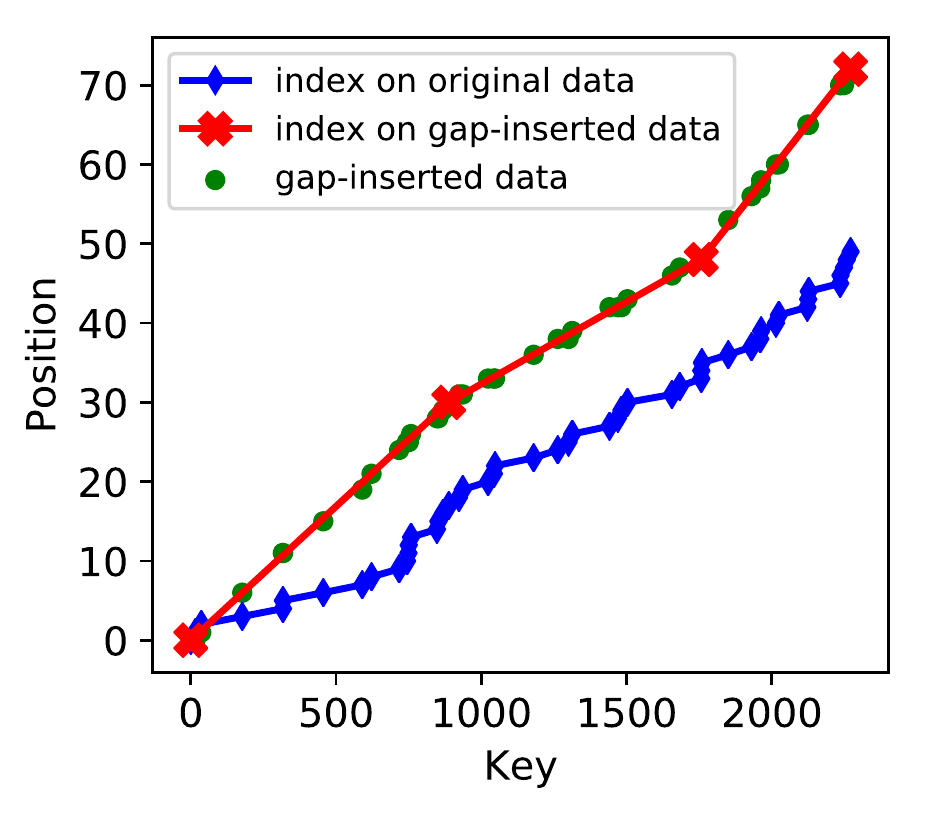}
		\label{fig:gap_demo}
		\vspace{-0.12in}
		\end{minipage}
	}
	\vspace{-0.2in}
	\caption{Examples of learned index with sampling and gap insertion.}
\end{figure}
\section{Learning Index with Gap Insertion}
\label{sec:gap}

\subsection{Gap insertion}
As discussed above, the proposed sampling technique can somehow improve the generalization ability of learned index methods. Here we explore more along this direction: Can we further enhance the generalization ability of learned index methods and improve their preciseness?  
In the fruitful research work of index structure in the database community, there are some ideas about how to leave certain empty space for dynamic data, such as Packed Memory Array \cite{bender2007adaptive} or Gapped Array \cite{ding2019alex}. 
Actually, we can further leverage the empty space to re-distribute the data such that the updated key-position distribution is easy to be learned and the generalization ability of index can be further enhanced.

Let's consider the example in Figure \ref{fig:gap_demo}: Can we insert some gaps, move the original data (\ie, blue dots) to the gap-inserted data (\ie, green dots), and learn the index as red segments? 
If so, the newly learned index (\ie, red segments) fits the gap-inserted data better and has better generalization ability. 
However, we should note that the number of inserted gaps cannot be very large since the gaps can increase storage and query costs, which are also part of the optimization goals. 

\paragraph{Result-Driven Gap Insertion}.
The gap insertion can be formulated as a \textit{position manipulation} problem: given a set of data $\data = \{x_i, y_i\}_{i=1}^{n}$ with the monotonicity of key-position, how to insert gaps and adjust each record's position from $y_i$ to $y_i^g$ while preserving the monotonicity of key-position, such that the preciseness of learned mechanism can be improved? 
Formally, let $\data_g=\{x_i, y_i^g\}_{i=1}^{n}$ be a gap-inserted data and $u_i=y_i^g-y_{i-1}^g-1$ be the number of inserted gaps between $x_{i-1}$ and $x_i$. Then $y_i^g=y_i+\sum_{1 \leq j \leq i}u_j$ and we aim to find a manipulated $\data_g$ that maximizes the objective:
\begin{equation}
\begin{aligned}
\max_{\{y_i^g\}}   &\sum_{\{x_i,y_i\}\in \data} L(y_i, M^*_{\data} (x_i)) - L(y^g_i, M^*_{\data_g}(x_i)) \\
\text{s.t.} ~ &\forall i:~ u_i \in \mathbb Z_{\ge 0}, ~~  \text{and} ~~ \sum_{i=1}^n u_i \leq \rho \cdot n ~,  \\
&\forall i: y_i^g=y_i+\sum_{1 \leq j \leq i}u_j, \, \, \forall i, j: y_i^g < y_j^g, ~~\text{iff} ~~  x_i < x_j,
\label{equ:gap-obj}
\end{aligned}
\end{equation}
where $\rho$ indicates the gap ratio that defines at most how many gaps (\ie, $\rho \cdot n$) can be introduced, $M^*_\data$ and $M^*_{\data_g}$ indicate the optimal learned index mechanism from original data $\data$ and gap-inserted data $\data_g$ respectively.

Inspired by block coordinate gradient descent \cite{bertsekas1997nonlinear}, here we describe a two-step \textbf{result-driven} solution to solve the optimization problem in Equation (\ref{equ:gap-obj}): The solution first proposes a new index mechanism  $M'$ that can achieve better preciseness than the index mechanism on the original data, and then move records to $y^g$ that are as close as possible to the positions predicted by $M'$. 
In other words, as the new index mechanism $M'$ meets our optimization goal, i.e., reducing the prediction error, we can ``backward inference'' $y^g$ using the predictions from $M'$ and learn a mechanism $M_{\data_g}$ from the estimated gap-inserted data $\data_g$.

\begin{figure}[htbp]
	\centering
	\subfigure{
		\begin{minipage}[]{0.48\textwidth}
			\centering
			\includegraphics[width=\textwidth]{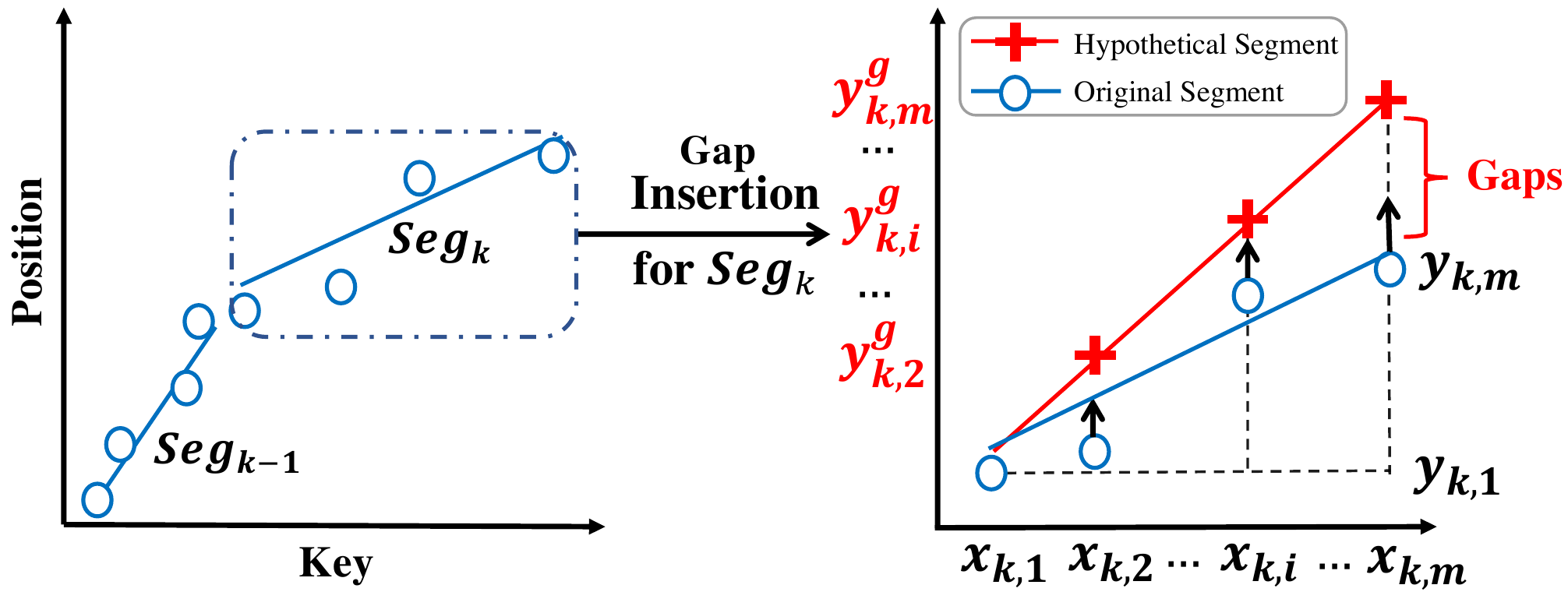}
		\end{minipage}
	}
	\vspace{-0.2in}
	\caption{The illustration of result-driven gap insertion.}
	\label{fig:gap_insert_demo}
\end{figure}

Specifically, since the linear models are adopted in existing learned index methods including RMI, FITing-Tree and PGM, here we discuss how to manipulate the positions for linear models based on the proposed result-driven solution. 
Considering that real-world datasets usually cannot be indexed with only one linear model, we first learn an index mechanism with $K$ linear segments to globally split the data, then locally insert gaps for each linear segment.
Let's consider the example in Figure \ref{fig:gap_insert_demo}: for the $k$-th linear segment ($k \in [1,K]$), we can propose a hypothetical linear model shown in red line by connecting two anchoring points after gap insertion: the first data point $\{x_{k,1}, y_{k,1}^g\}$ and the last data point $\{x_{k,m}, y_{k,m}^g\}$. For the first data point, $y_{k,1}^g = y_{k,1} + \sum_{j<=k}U_j$ where $U_j$ is the number of total inserted gaps in the previous $j$-th segment, while for the last data point, $y_{k,m}^g = y_{k,m} + \sum_{j<=k}U_j + \rho (y_{k,m}-y_{k,1})$ as $\rho (y_{k,m}-y_{k,1})$ is the total number of gaps to be inserted in the $k$-th segment.
In this way, the gap-inserted positions $y_{k,i}^g$ can be calculated as:
\begin{equation}
\label{equ:gap-insert}
\begin{split}
y_{k,i}^g& =y_{k,1}^g + (x_{k,i}-x_{k,1}) \times \frac{y^g_{k,m}-y_{k,1}^g}{x_{k,m}-x_{k,1}}, \\
&=y_{k,1} + \sum_{j<=k}U_j +  (x_{k,i}-x_{k,1}) \times \frac{(y_{k,m}-y_{k,1})(1+\rho)}{x_{k,m}-x_{k,1}}. \\
\end{split}
\end{equation}
From a geometric perspective, the data after inserting gaps are placed along the red hypothetical line and consequently, the data is re-distributed to meet our ``expected linear results''.
Note that the estimated $y^g$s can be non-integers and we will discuss how to physically place the keys according to their $y^g$s in Section \ref{sec:place-strategy}.
Now let's first theoretically examine the effectiveness of the proposed result-driven gap insertion technique.

\paragraph{Theoretical Analysis}.
In this section, we will theoretically show that the index mechanism learned on $\data_{g}$ achieves better preciseness and tighter generalization bound than the one learned on $\data$. 
Here we leverage the information bottleneck principle \cite{shamir2010learning} to quantify the mutual information between the input keys and the positions to be predicted, and analyze the preciseness and generalization for indexes learned on $\data_{g}$ and $\data$. 

Let's consider the index learning task in the context of information compression: The key $x \in \mathcal{X}$ and the position $y \in \mathcal{Y}$ are both random variables, and there are some statistical dependencies between $x$ and $y$. 
The index learning task can be regarded as to find an optimal compact representation $\tilde{x}$ of $x$, which compresses $x$ by removing the un-related parts of $x$ that have no contribution to the prediction of $y$; and then learn the correlation between $\tilde{x}$ of $y$.
Formally, the index learning task is minimizing $I(x;\tilde{x})-\beta I(\tilde{x};y)$, where $I()$ indicates the mutual information and $\beta$ is a coefficient to adjust the trade-off between the degree of compression $I(x;\tilde{x})$ and the degree of preserving predictive information $I(\tilde{x};y)$. 

Let $\mathcal{\tilde{X}}$ be the set of minimal sufficient statistics of $\mathcal{X}$ with respect to $\mathcal{Y}$, \ie, $\forall x \in \mathcal{X}, \tilde{x} \in \mathcal{\tilde{X}}, y \in \mathcal{Y}: p(x | \tilde{x}, y)=p(x | \tilde{x})$. 
In the context of index learning, we interpret $\mathcal{\tilde{X}}$ as the optimal hidden compact representation of keys $\mathcal{X}$. 
Note $\mathcal{\tilde{X}}=\mathcal{X}$ is always sufficient for $\mathcal{Y}$ while usually not the optimal compact one with minimal $|\mathcal{\tilde{X}}|$.
Then we can see that \textit{$|\mathcal{\tilde{X}}|$ becomes smaller on $\data_{g}$ after inserting gaps into the original data $\data$}, which makes the index learning easier and improves the index preciseness. 
To clarify the reason, recall that we insert gaps based on several hypothetical linear models.
As the linearity of transformed data increases, the linear correlation between keys and gap-inserted positions increases accordingly.
In other words, we need less data information to learn the parameters of linear models, resulting in redundant $\mathcal{\tilde{X}}$ and smaller $|\mathcal{\tilde{X}}|$.
In the extreme case, $\mathcal{\tilde{X}}$ can have only one trivial $\tilde{x}$, that is $\mathcal{\tilde{X}}=\{x_1-1\}$, if the whole transformed data is on a simple line (e.g., $y=x-x_1+1$).

Besides, we can show that the index learned on $\data_{g}$ achieves a tighter generalization bound after inserting gaps into the original $\data$.
Let $\hat{I}()$ be the \textit{empirical estimate} of the mutual information for a given data, which is a sample of size $n$ from the joint distribution $\{\mathcal{X}, \mathcal{Y}\}$. 
Shamir et al. \cite{shamir2010learning} prove that the generalization error $E(\tilde{x},y)=|I(\tilde{x};y) - \hat{I}(\tilde{x};y)|$ can be bounded in a data-dependent form with a probability of at least $1-\delta$:
\begin{equation}
\begin{split}
\small
&E(\tilde{x},y) \leq ~ \frac{(3|\mathcal{\tilde{X}}| +2) \sqrt{log(4/\delta)}}{\sqrt{2n}} 
+ \frac{( |\mathcal{Y}|+1)(|\mathcal{\tilde{X}}| +1)-4}{n}.
\end{split}
\end{equation}
In summary, the generation error between the optimal hidden compact representation and the empirical estimates from finite sample $\data$ or $\data_g$ is bounded in $O(\frac{|\mathcal{\tilde{X}}| |\mathcal{Y}|}{\sqrt{n}})$. 

Now let's compare the bounds on the original data $\data$ and the gap-inserted data $\data_{g}$. 
Note that we transform the original $y$ into $y^g$ with a one to one mapping due to the key-position monotonicity, and thus the $|\mathcal{Y}|$ and data size $n$ remain the same. 
The other factor $|\mathcal{\tilde{X}}|$ does matter, and it becomes smaller after gap insertion as we analyzed above.
As a result, the index mechanism learned on $\data_{g}$ has a tighter generation bound than the index mechanism learned on $\data$, and it can generalize to possible keys better. 
Later, we will discuss how the proposed gap insertion technique can handle dynamic scenarios in Section \ref{sec:handle-dynamic} and experimentally confirm its advantages in Section \ref{exp:gap}. 

\paragraph{Gap Insertion for Non-Linear Models}.
So far, our discussion about the proposed gap insertion is based on a collection of linear models. 
However, the idea of results-driven gap insertion is general and easy to be extended to other non-linear models.
Specifically, as long as the non-linear models to be learned are monotonically increasing functions, we can also introduce the non-linear hypothetical lines by anchoring a few points whose transformed positions can be determined.
Then based on the hypothetical lines, we can inference the positions for the keys of other non-anchoring points to fit hypothetical models.

\subsection{Physical Implementation of Gaps}
\label{sec:place-strategy}
\paragraph{Physical Key Placement}.
We have discussed how to logically estimate the gaps to be inserted such that a better index can be learned on the gap-inserted data. 
However, as the estimated positions can be non-integral, we need to round them into integers and physically place the keys according to their adjusted positions.
This is non-trivial since we have to maintain the key-position monotonicity during the placement.
Let's assume two anchoring keys $x_{i}$ and $x_{j}$ where $x_{i}<x_{j}$, whose physical positions are integers and have been determined to be $y^g_{i}$ and $y^g_{j}$ respectively. In practice, the first and last points of each learned segment can be such anchoring points (please refer to \emph{Result-Driven Gap Insertion} above).
Consider $m$ non-anchoring keys to be placed, whose keys are larger than $x_{i}$ and less than $x_{j}$, we need to choose suitable positions to place them.
On one hand, the rounding of estimated (non-integral) positions may lead to conflicted $y^g$s. 
On the other hand, there may be more than $m$ available positions used for physical placement, i.e., $y^g_{j}-y^g_{i}>m$.

\begin{figure}[htbp]
	\centering
	\subfigure{
		\begin{minipage}[]{0.47\textwidth}
			\centering
			\includegraphics[width=\textwidth]{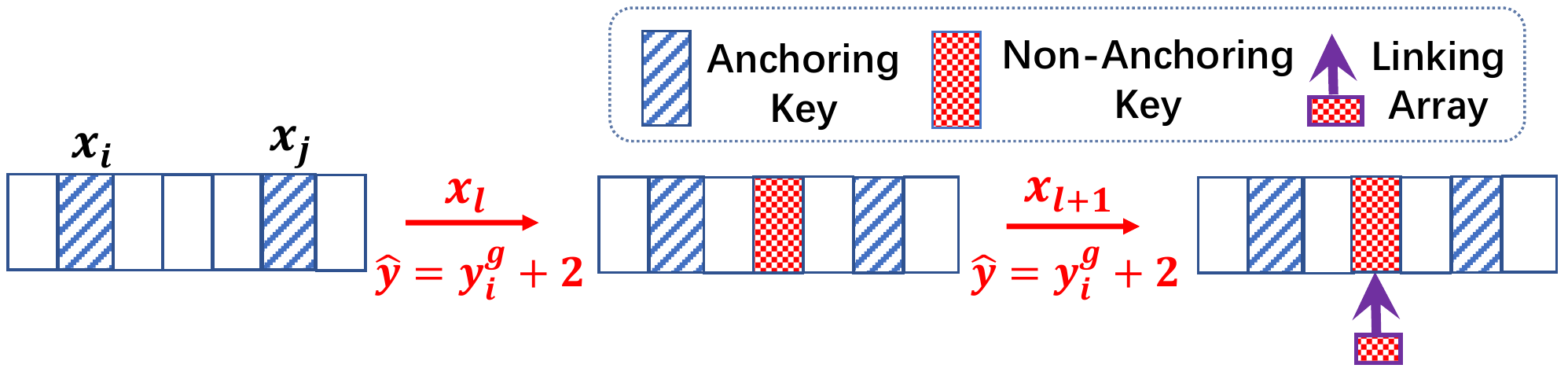}
		\end{minipage}
	}
	\vspace{-0.2in}
	\caption{Illustration of linking array based key placement.}
	\label{fig:key_place_demo}
\end{figure}

To address this problem, we propose a \textit{linking array} based key placement strategy as shown in Figure \ref{fig:key_place_demo}. 
We place a key $x_l$ according to its predicted position $M(x_l)$. If two keys get conflicted predictions, we place them on an external linking array with the same position. 
This strategy reduces the differences between physical positions and result-driven estimated positions by fully utilizing the prediction ability of learned index $M$, at the price of introducing potential disambiguation costs to determine exact positions in the linking arrays, and increasing storage costs for linking arrays.
We also explore two other strategies to balance the advantage of accurate prediction and the disadvantage of extra cost, and experimental results show that the strategy presented above works well in practice.
Due to the space limitation, we omit the details and keep with the above linking array based key placement strategy.

As a result, we place all data in a gapped array $G$ and several linking arrays $\mathcal{A} = \{A_i\}$. A linking array $A_i$ contains the $i$-th occupied key in $G$ and at least one other key having the same position $y^g=i$ in $G$. The $i$-th occupied key in $G$ will be the minimum keys of $A_i$, i.e., $G(i) = min(A_i)$. Clearly, the key-position monotonicity is maintained for keys in the first-level gapped array.
Next, we will describe how to read (\ie, lookup operation) and write (\ie, dynamic scenario) with the gap-based array.

\paragraph{Lookup Operation}.
As mentioned above, in the gapped array $G$, there are more positions than the number of indexed keys, \ie, $y^g_{j}-y^g_{i}>m$. Further, the proposed linking arrays can increase the number of unoccupied positions. So how to conduct lookup operations on such data layout? The key idea here is to maintain a \textit{total order} for keys in the first-level gapped array $G$, such that the unoccupied positions in $G$ are comparable to indexed keys. 
Specifically, we set the key of unoccupied position to be the key of the first occupied position at its right hand, and use an additional indicator to indicate that such unoccupied positions have no payloads. 
In this way, we can define a comparison rule for the situation of having the same key but different payloads: a key pointing to an empty payload is smaller than the same key pointing to a non-empty payload, and thus the $G$ is totally ordered. 

In practical, if the searched positions point to the secondary linking arrays, we conduct a linear scan on the linking arrays due to the fact that there are usually a limited number of keys with the same conflicted predicted position.

\subsection{Dynamic Scenario}
\label{sec:handle-dynamic}
\paragraph{Handling Dynamic Scenarios}.
Dynamic operations, especially inserting new keys, are challenging for learned index since the machine learning models may need to be re-trained to maintain precise predictions. 
Thanks to the proposed gap insertion, we can introduce a simple yet effective extension to efficiently support dynamic operations, meanwhile, maintaining comparable prediction preciseness without model re-training.
This is because we have re-distributed the data by introducing gaps in a result-driven manner. 
Further, we place keys directly according to their predicted positions and allow position conflicts, thus the unoccupied positions are data-dependently reserved for possible inserted new keys.
In other words, we can insert new keys based on their predicted positions by a learned index $M$. These positions can be either unoccupied or occupied, and for both cases, the inserted data follow hidden key-position distribution already learned by $M$, such that $M$ can maintain preciseness at no price or a small price of correction cost in the lookup operation.

Specifically, given a key $x$ to be inserted, we expect to place $x$ on the position as close as possible to its predicted position by $M$. 
Meanwhile, to simplify the lookup operation after new data insertion, we need to maintain the total order for $G$ with corresponding linking arrays $\mathcal{A}$: ${\forall}A_{i-1} \in \mathcal{A},~~  G(i-1)=min(A_{i-1}) \leq max(A_{i-1}) < G(i)$. 
Using $M$, we first get $x$'s predicted position $\hat{y}$ and its upper bound's position $y_{ub}$ in $G$, i.e., the position of largest key in $G$ that is less than or equal to $x$, and then we insert $x$ in either $G(\hat{y})$ (unoccupied case) or the linking array $A_{y_{ub-1}}$ (occupied case) to maintain the key-position monotonicity and the total order of $G$.

For the delete operation, first, we look up the $x$ to be deleted. If $x$ is stored in a linking array $A$, we can easily remove $x$ from $A$ when $|A|>2$, or make the corresponding position an occupied one with the other key in $A$, and delete $A$ when $|A|=2$. Otherwise, the deletion will make an unoccupied position, and we need to update the keys in $G$ whose values are the same as $x$, by setting them as the key at $x$'s upper bound's position $y_{ub}$, i.e., $\forall y_j: G(y_j)=x$, set $G(y_j) \leftarrow G(y_{ub}) $. 
For the update operation, we can look up the data to be updated using its key and simply reset its payload value.

\subsection{Reducing Cost of Gap Insertion}
Now let's analyze the cost of the proposed gap-based index learning.
Recall that we first learn an index with $K$ linear segments on original data $\data = \{x_i, y_i\}_{i=1}^{n}$, then we generate gap-inserted data $\data_g = \{x_i, y_i^g\}_{i=1}^{n}$ and re-learn a better index on $\data_g$. With existing learned index methods such as FITing-Tree \cite{galakatos2019fiting} and PGM \cite{Ferragina:2020pgm}, the cost of these three steps are all $O(n)$.
Although the complexity of the proposed gap-based index learning is still $O(n)$, we introduce some extra training costs due to the inserted gaps.
Fortunately, we can leverage the sampling technique mentioned in Section \ref{sec:sample} to further reduce the learning cost using small sampled data, and thus maintain high learning efficiency. 

\paragraph{Combining Sampling and Gap Insertion}.
Formally, with a sample rate $s \in (0, 1.0 ]$ and the sample size $n_s=n*s$, we first learn index $M$ on a sampled subset $\data_s=\{x_i, y_i\}_{i=1}^{n_s}$ of $\data$. 
Then we use all the $K$ segments of $M$ to estimate the positions of anchoring keys, physically place other non-anchoring keys based on their predicted positions from proposed hypothetical models, and get the gap-inserted dataset $\data_{s, g}=\{x_i, y_i^g\}_{i=1}^{n_s}$.
Finally we re-learn an index mechanism on $\data_{s, g}$, and get the whole gap-inserted data $\data_{g}=\{x_i, y_i^g\}_{i=1}^{n}$ by physically placing the un-sampled keys $\data_{s}^-=\data-\data_s$ into $\data_{s, g}$. 
Summing up, we can learn an index with a data scale of $n_s$ and physically maintain the whole data layout in $O(n)$.
In Section \ref{exp:effect_of_gap}, we will experimentally show that by combining sampling and gap insertion techniques, both the preciseness and efficiency of existing learned index methods can be significantly improved.
\section{Experiments}
\label{sec:exp}
In this section, we conduct experiments aiming to answer the following questions:
(1) What are the strengths and weaknesses of the existing learned index methods evaluated by the proposed MDL framework (Section \ref{exp:trade-off-MDL})?
(2) Can we improve the learning efficiency of the learned index by the sampling technique (Section \ref{exp:sample})?
(3) Can we enhance the learning effectiveness of learned index by the proposed result-driven gap insertion technique (Section \ref{exp:gap})?

\subsection{Experimental Settings}
We conduct all the experiments on a Linux server with an Intel Xeon Platinum 8163 2.50GHz CPU, whose L1 cache size, L2 cache size and L3 cache size are 32KiB, 1MiB and 33MiB respectively. 

\paragraph{Baselines}.
We compare a traditional and three learned index methods, which are also adopted as base methods to incorporate the plug-in sampling and gap insertion techniques.

\textit{B${+}$ Tree}: we use a standard in-memory B+ Tree implementation, stx::btree (v0.9)  \cite{stxbtree}. Following \cite{galakatos2019fiting, kraska2018case}, we evaluate the B+ Tree index with dense pages, \ie, the filling factor = 100\%.

\textit{Recursive Model Index (RMI)} \cite{kraska2018case}: RMI is a hierarchy learned index method consisting of typically two-layer or three-layer machine learning models. Following previous works \cite{kraska2018case, Ferragina:2020pgm, ding2019alex}, we adopt a two-layer RMI with linear models.

\textit{FITing-Tree} \cite{galakatos2019fiting}: This is an error-guaranteed learned index method using a greedy shrinking cone algorithm to learn piece-wise linear segments. The learned segments are organized by a B+ Tree and here we adopt the stx::btree to organize these learned segments.

\textit{Piecewise Geometric Model (PGM)} \cite{Ferragina:2020pgm}: PGM is a state-of-the-art error-guaranteed learned index method, which improves FITing-Tree by learning an optimal number of linear segments. There are three PGM variants based on binary search, CSS-Tree\cite{rao1998csstree} and recursive construction. Here we evaluate its recursive version since it beats the other two variants.

\paragraph{Datasets}. 
We conduct experiments on four wildly adopted real-world datasets that cover different data scales, key types, data distributions and patterns:

\emph{Weblogs} \cite{kraska2018case, galakatos2019fiting, Ferragina:2020pgm}:  The Weblogs dataset contains about 715M log entries requesting to a university web server. The index keys are unique log timestamps. This dataset contains typical non-linear temporal patterns caused by school online transactions, such as department events and class schedule arrangements.

\emph{IoT} \cite{galakatos2019fiting, Ferragina:2020pgm}: The IoT dataset contains about 26M recordings from different IoT sensors in a building. The index keys are unique timestamps of the recordings. This dataset has more complex temporal patterns than Weblogs, since IoT data are more diverse  (e.g., motion, door, etc.) and prone to noise during the data collection.

\emph{Longitude and LatiLong}
\cite{kraska2018case,galakatos2019fiting,ding2019alex,Ferragina:2020pgm}: These two datasets contain location-based data that are collected around the world from \textit{Open Street Map} \cite{OpenStreetMap}. 
The index keys of \emph{Longitude} are the longitude coordinates of about 1.8M buildings and points of interest.
Similar to \cite{ding2019alex}, the index keys of \emph{LatiLong} is compounded of latitudes and longitudes as $key = 90 \times latitude + longitude$.

\paragraph{Evaluation Metrics}.
For the storage cost evaluation, we measure the \textit{index size}. We use 64-bit payloads for all baselines and 64-bit key pointers for all datasets.
The index size of B+ Tree is the sum of the sizes of inner nodes and the sizes of leaf nodes including payloads.
The index size of RMI is the sum of payloads and the sizes of linear models, including slopes, intercepts, maximum positive/negative prediction errors storing as double-precision floats.
The index sizes of Fitting-Tree and PGM are the sum of payloads and the sizes of their linear segment models, including slopes and intercepts storing as double-precision floats.

For the efficiency evaluation, we measure several kinds of time costs in nanoseconds, including the \textit{index construction time}, the \textit{index prediction time} per query (\ie, getting predicted position $\hat{y}$ given queried key $x$), the \textit{index correction time} per query (\ie, getting the true position $y$ given $\hat{y}$) and the \textit{overall query time} per query (\ie., getting $y$ given $x$). 
Besides, we calculate the \textit{Mean Absolute Error (MAE)} between predicted positions and true positions as $\frac{1}{|D|} \sum_{x \in D}|y - \hat{y}|$). MAE is a metric widely adopted for machine learning algorithms and determines the index correction time in the context of learned index.

\subsection{Comparison under MDL Framework}
\label{exp:trade-off-MDL}
The proposed MDL-based framework quantifies an index as terms of $L(M)$ and $L(\data|M)$.
In this subsection, we compare several existing index methods to demonstrate several performance trade-offs between $L(M)$ and $L(\data|M)$, and the impact of varying the mechanism family $\mathcal{M}$. We demonstrate the results on IoT dataset and omit the results on the other three datasets due to the space limitation and similar conclusions.

\subsubsection{Performance Trade-Off}
From the view of machine learning, $L(\data|M)$ measures the prediction loss on training data while $L(M)$ plays a regularization role of the learned model.
These two terms are usually contradicted and can be balanced by the coefficient $\alpha$ that controls performance trade-offs.
As analyzed in Section \ref{sec:alpha-in-MDL}, several parameters of existing index methods implicitly take the role of $\alpha$: the number of the layer-2 models in RMI is proportional to $\alpha$, while the page size of B+ Tree and error bound $\epsilon$ of Fitting-Tree and PGM are inversely proportional to $\alpha$.
Here we vary these tunable parameters to study various performance trade-offs by adopting different $L(M)$ and $L(\data|M)$.

\begin{table*}[tb]
	\centering
	\begin{tabular}{ccccccc}    
		\toprule      
		 & $T_{build}$& $T_{predict}$& $T_{correct}$ & $T_{overall}$ & Index Size & MAE \\    
		\midrule
		B+ Tree (pageSize=256) & 7,770,595,420 & 305 & 338 & 643 & 489,877,504 & 63.5005 \\
		RMI ($|M|$=100k) & 682,842,396 & 68 & 539 & 607 & 128,720,736 & 173.513 \\
		FITing ($\epsilon=128$, $|M|$=11,830) & 788,580,446 & 106 & 203 & 309 & 121,910,880 & 27.3392 \\
		PGM ($\epsilon=128$, $|M|$=8,813) & 1,556,264,268 & 121 & 224 & 355 & 121,740,264 & 36.817 \\
		\bottomrule
	\end{tabular}
	\caption {Performance comparison on the IoT dataset for B-tree with dense page, RMI, Fitting-tree and PGM. For RMI, FITing-Tree and PGM, $|M|$ indicates their numbers of last-layer linear models. $T$ indicate time in ns, and the index size is in bytes.} 
	\label{tab:model-capacity-compare} 
	\vspace{-0.2in}
\end{table*}

\paragraph{Trade-off between Storage Cost and Query Efficiency.} 
We first set $L(\data|M)= t_q(\data|M)$ indicating the \textit{overall query time} per query and $L(M)=SIZE(M)$ indicating the \textit{index size} to explore the trade-off between storage cost and query efficiency. 
We plot the curves of overall query time per query and index size by varying $\alpha$s of different methods in Figure \ref{fig:tradeoff_space_query}. 

\begin{figure}[htbp]
	\centering
	\subfigure{
		\begin{minipage}[]{0.47\textwidth}
			\centering
			\includegraphics[width=\textwidth]{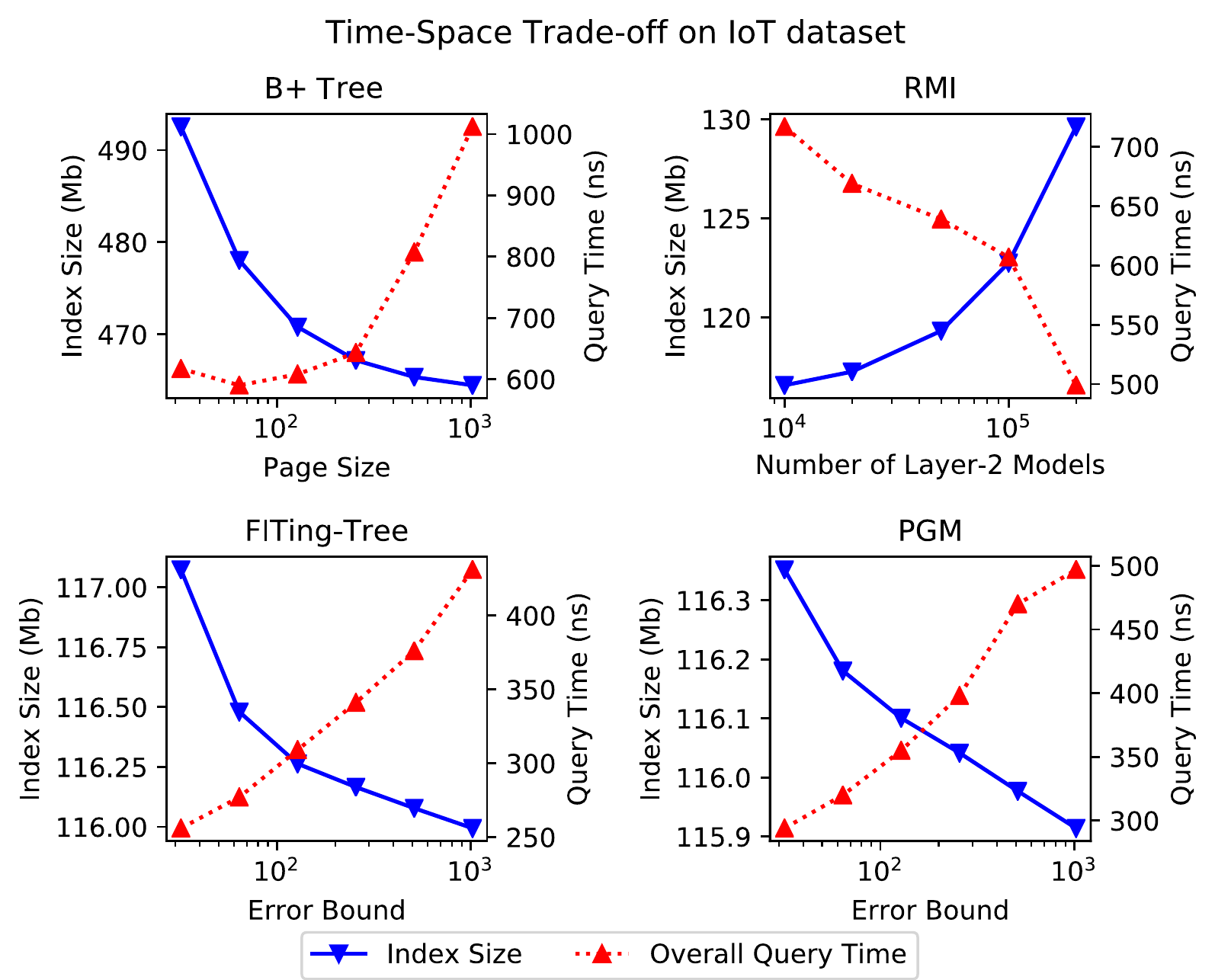}
		\end{minipage}
	}
	\vspace{-0.2in}
	\caption{Trade-off between storage cost and query efficiency of four index methods. }
	\label{fig:tradeoff_space_query}
\end{figure}

Overall speaking, all the four methods have consistent trade-off trends: a smaller $\alpha$ achieves a smaller index size, but leads to larger query time, as it penalizes less on the term $L(\data|M)$. 
Here we test several $\alpha$s for each method, and among them, some choices gain good trade-offs (e.g., $\epsilon=128$ for FITing-Tree). 
These results naturally raise an open question: what is the ``best'' space-time trade-off and how to achieve it by tuning $\alpha$?
In practice, it depends on the demand of users.
PGM has explored both ends of the space-time trade-off by varying the degree of linear approximation and searching an $\epsilon$ that achieves minimal space or minimal query time. 
In the context of MDL, we can explicitly incorporate the preference of users into the objective functions. And then the fruitful ideas of hyper-parameter optimization \cite{feurer2019hyperparameter} can be incorporated to automatically search an optimal $\alpha$ to achieve preferable trade-offs.

Also note that the query time of B+ Tree increases heavily from page size 256 to 512 since the index begins to fall out of the L1-cache. 
Inspired by this phenomenon, we can see an interesting future direction: designing cache-sensitive regularization in index learning. For example, we can harness the information about the size of cache lines into the objective function by gate-control techniques \cite{Chung2014EmpiricalEO}.

\paragraph{Trade-off between Prediction and Correction.} 
We then drill down the index querying process and explore the trade-offs between prediction cost and correction cost. To measure the prediction cost, we set $L(M) = t_p(M)$, where $t_p(M)$ means the \textit{prediction time}. To measure the correction cost, we can set $L(\data|M)$ to be the \textit{correction time}, \ie $L(\data|M) = t_c(\data|M)$, or \textit{MAE}, \ie, $L(\data|M)= \frac{1}{|\data|}\sum_{x \in \data}|y - \hat{y}|$.

\begin{figure}[htbp]
	\centering
	\subfigure{
		\begin{minipage}[]{0.47\textwidth}
			\centering
			\includegraphics[width=\textwidth]{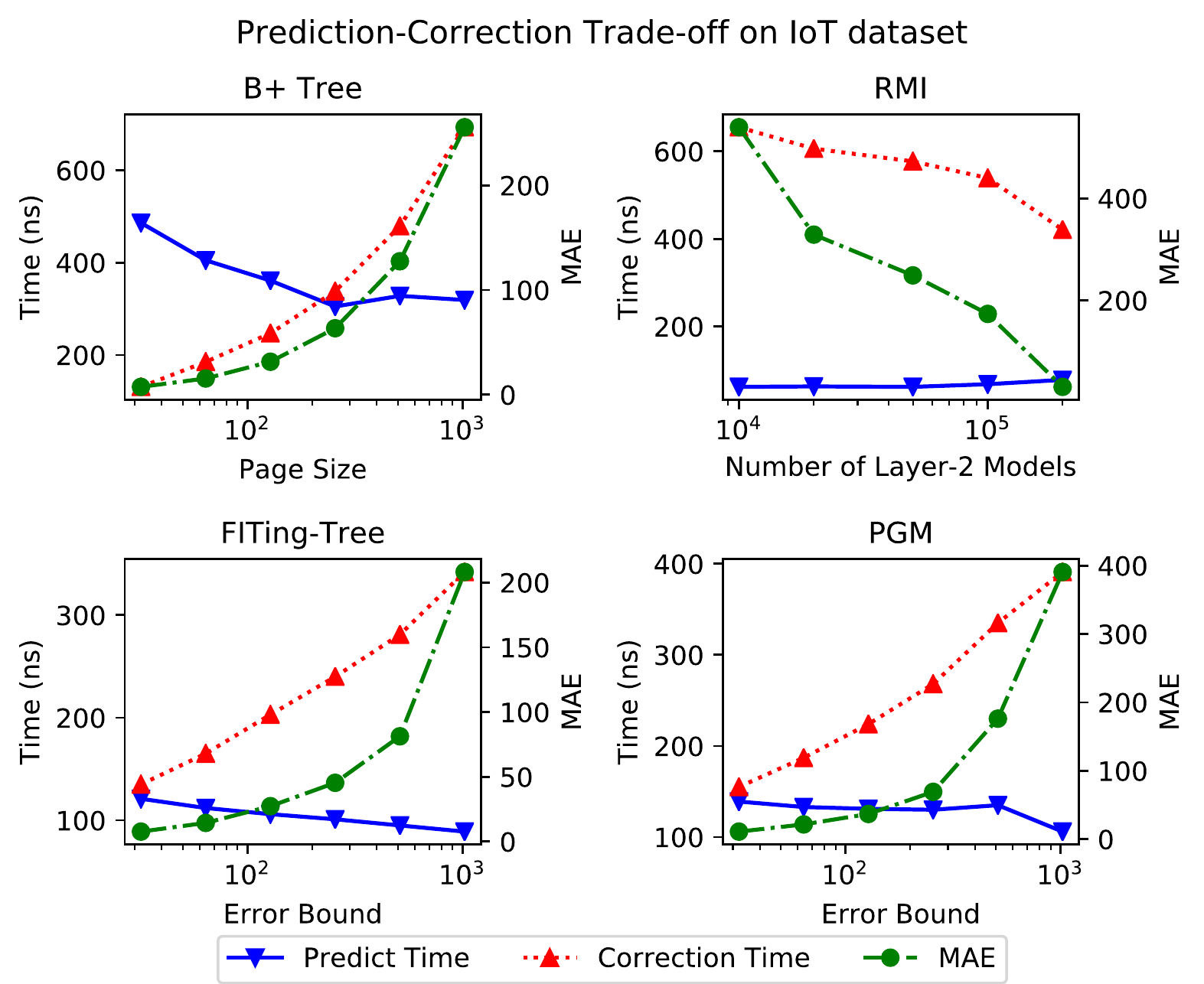}
		\end{minipage}
	}	
	\vspace{-0.2in}
	\caption{Trade-off between prediction time and correction cost of four index methods.}
	\label{fig:tradeoff_predict_correct}
\end{figure}

We plot the prediction time and correction cost (including \textit{correction time} and \textit{MAE}) for different index methods in Figure \ref{fig:tradeoff_predict_correct}. 
As the parameter $\alpha$ increases (\ie, page size and error bound decrease as they are inversely proportional to $\alpha$), the prediction time of B+Tree, FITing-Tree and PGM both increase while their correction costs decrease. 
From the view of MDL, a more complex $L(M)$, requiring a longer prediction time, usually make precise predictions and thus smaller correction cost $L(\data|M)$.
Here RMI is a bit different as its prediction time keeps almost the same as $\alpha$ changes.
This is due to the fact that the major part of inference in RMI is calculating $L$ linear functions through a $L$ layer tree (here $L$ is 2), which is independent of the tunable number of the layer-2 models.

Besides, the results in terms of \textit{MAE} and \textit{correction time} of all the index methods have the consistent increasing trends: all of them increase as $\alpha$ increases, while with different increasing speeds (see the different gaps between red lines and green lines). 
This meets our expectation since the evaluated index methods use a binary search within certain ranges: page size for B+ Tree, maximum positive/negative errors for RMI, and $\epsilon$ for FITing-Tree and PGM. Meanwhile, the MAEs are bounded by these values and approximatively reflect the correction time. 
Note that the MAE are usually smaller than these fixed search ranges, which allow us to further use exponential search to speed up the correction as studied in \cite{ding2019alex}. 

Another observation is that the prediction time is usually larger than the correction time for B+Tree, FITing-Tree and PGM. Comparing to their lookups within a continuous memory in the correction stage, their recursive lookups in the prediction stage tend to incur more costs, which are due to the organization of hierarchical models. This inspires us to design learned index having as few layers as possible, which is adopted by RMI method.

\subsubsection{The Effect of Model Capacity}
\label{exp:effect_of_model_capacity}
We have studied the impact of varying the regularization coefficient $\alpha$ for different index methods. Now we study the intrinsic property of the mechanism by comparing different model capacity, or say, comparing the different families of learning models. 

To fairly compare these four index methods, we choose the most favorable $\alpha$s for each of them, which are the values of $\alpha$s closest to the intersections of their time-space trade-off lines in Figure \ref{fig:tradeoff_space_query}. 
The results are summarized in Table \ref{tab:model-capacity-compare}. We can see that all the three learned index methods achieve much smaller storage costs while faster lookup speeds than traditional B+Tree.
Among the three learned index methods, we observe that FITing-Tree and PGM have much smaller MAE, better query efficiency, and smaller index size than RMI. 
However, RMI method requires less construction time since it doesn't need the organization of learned segments, such as an assistant B+Tree for FITing-Tree and a recursive strategy for PGM. 
Note that with the same $\epsilon=128$, PGM achieves a smaller index size than FITing-Tree by learning the optimal number of linear segments, while FITing-Tree adopts a greedy learning algorithm that leads to more segments than PGM.

\subsection{Learned Indexes with Sampling}
\label{exp:sample}
Recall that existing learn index methods including RMI, FITing-Tree and PGM need to scan at least one pass of the whole data to learn several sub-models, and further organize the learned sub-models by a model tree or a B+Tree.
There is a great potential to reduce the data scanning cost with smaller data, and reduce the index organization cost with fewer sub-models to be learned.
Here we plug the proposed sampling technique into these learned index methods to accelerate the index learning. 
Specifically, given a sample rate $s \in (0, 1.0]$, we first get a sampled dataset $\data_s$ by randomly sampling $s \times |\data|$ keys from $\data$, then learn the indexes from $\data_s$, and finally test the indexes on the whole dataset $\data$. 
When directly applying this procedure with original RMI, FITing-Tree and PGM methods, we found some pretty large prediction errors caused by few unsampled keys that cannot be covered by learned sub-models (\ie, linear segments here).
Fortunately, we can eliminate these large errors by some simple yet effective patches. For FITing-Tree and PGM, we connect the adjacent segments learned from $\data_s$ to cover all the un-sampled keys. For RMI, we propose a \textit{RMI-Nearest-Seg} patch, which re-assigns a key covered by an un-trained (empty) sub-model to its nearest trained sub-model.
Since the sampling may cause few violations of error bounds, we adopt the exponential search to find the searching boundary around the predicted position. 

\begin{figure}[htbp]
	\centering
	\subfigure{
		\begin{minipage}{0.47\textwidth}
			\centering
			\includegraphics[width=\textwidth]{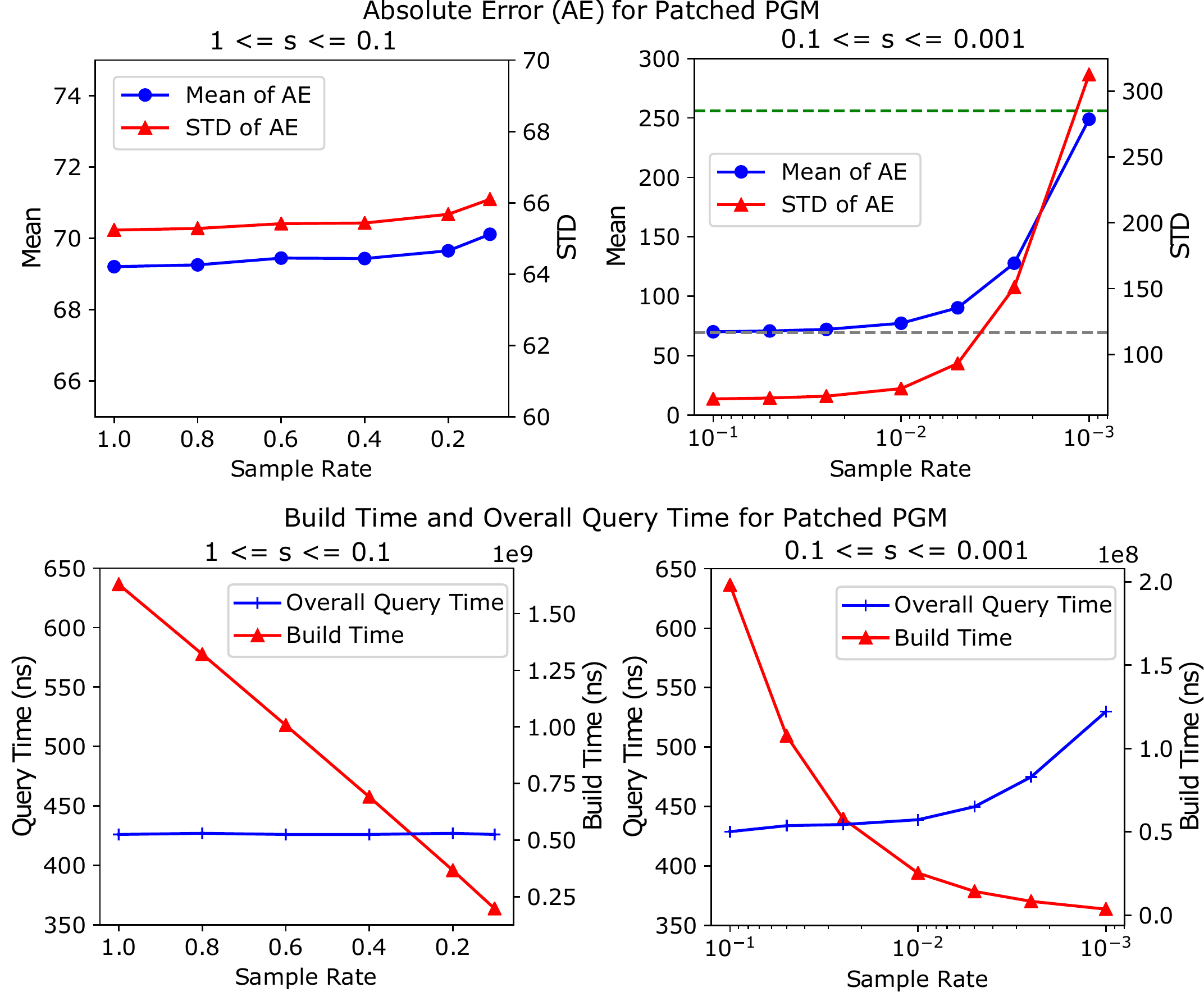}
		\end{minipage}%
	}
	\vspace{-0.2in}
	\caption{MAE, build time and overall query time of PGM index with varied $s$. The gray and green horizontal dotted lines indicate the MAE of $M_{\data}$ and $\epsilon=256$ respectively.}
	\label{fig:sample-overall-iot}
\end{figure}

By varying $s$ from 1.0 to 0.001, we evaluate the \textit{MAE}, \textit{build time} and \textit{query time} for all patched learned indexes on all the datasets.
The results of PGM on IoT dataset are shown in Figure \ref{fig:sample-overall-iot}, the other two index methods or on other datasets have similar conclusions and we omit those results due to the space limitation. 
In Figure \ref{fig:sample-overall-iot}, the two right subfigures show the results for extremely small sample rates (\ie, 0.1 to 0.001). 
For comparison, let's denote the index learned from the original whole data $\data$ as $M_{\data}$.
In the top-right subfigure, the gray and green horizontal dotted lines indicate the MAE of $M_{\data}$ and the adopted error bound $\epsilon=256$ respectively.

\paragraph{Construction Speedup}.
From Figure \ref{fig:sample-overall-iot}, we can see that the patched PGM gains significant construction speedup (e.g., the build time decrease from $1.63e9$ when $s=1$ to $1e8$ when $s=0.01$, which becomes about 78x smaller), while maintaining non-degraded query performance (e.g., when $s=0.01$, the MEA is still very close to the one of $M_{\data}$, the gray dotted line).  
Generally speaking, the build time linearly decreases as the sample rate decreases. 
On the other hand, the curves of MAE and query time are near-horizontal when the sample rate decreases until the very small one (e.g., $s=0.0025$).
These results show that we can significantly accelerate the learning of index with the proposed sampling technique while the learned index is still precise.

\paragraph{Generalization Improvement}.
As discussed in Section \ref{sec:sample-disscuss}, the sampling technique can improve the generalization ability of learned index methods, which leads to a fewer number of learned segments and correspondingly smaller index size. 
We conduct statistics of the number of learned segments for patched FITing-Tree and patched PGM, and the results are illustrated in Figure \ref{fig:sample-number-segs}. 

\begin{figure}[htbp]
	\centering
	\subfigure{
		\begin{minipage}[]{0.43\textwidth}
			\centering
			\includegraphics[width=\textwidth]{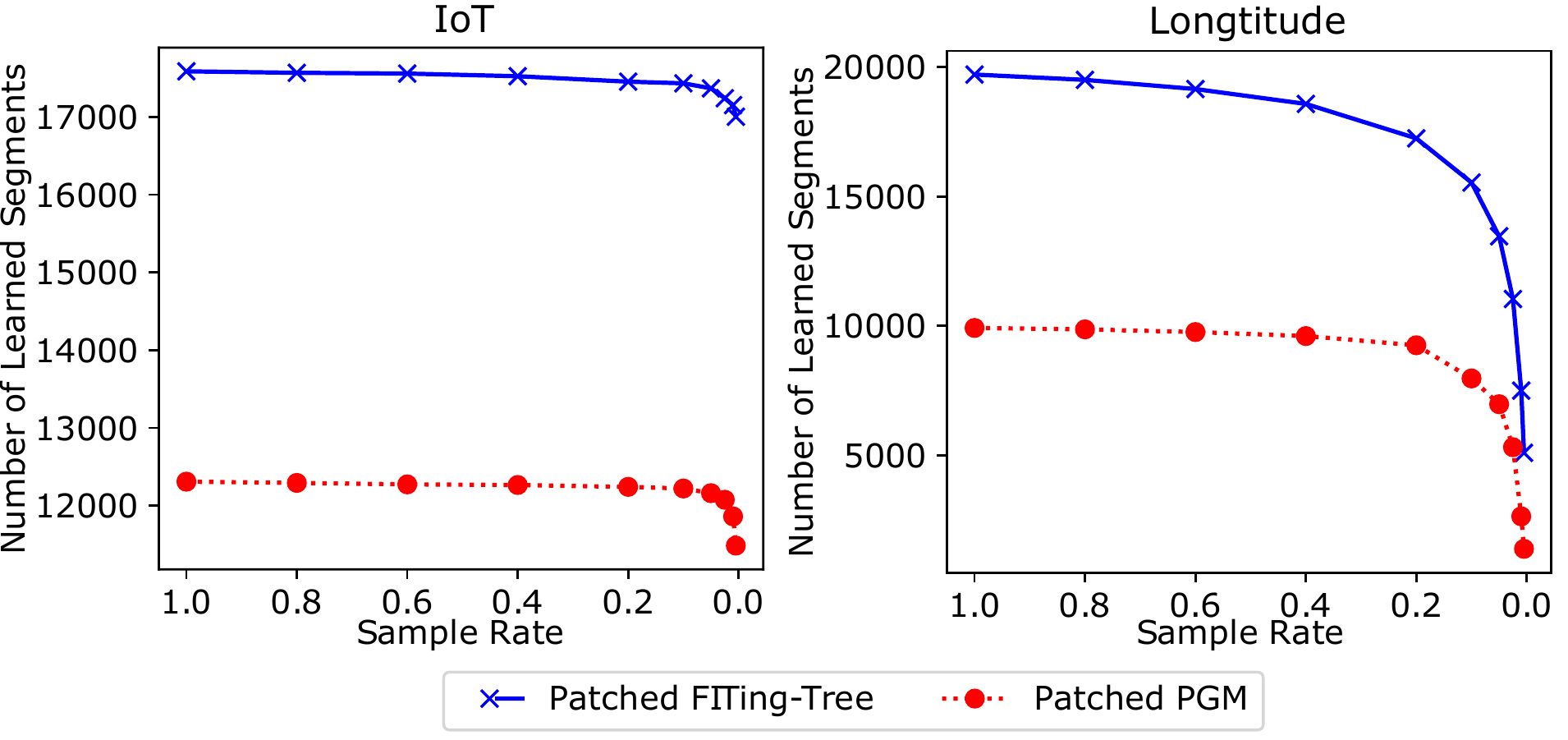}
		\end{minipage}%
	}%
	\vspace{-0.2in}
	\caption{Number of learned segments for patched FITing-Tree and patched PGM with varied $s$.}
	\label{fig:sample-number-segs}
\end{figure}

We can observe that the numbers of learned segments of both patched FITing-Tree and PGM decrease as sample rates decrease. 
This is because that as the sampled dataset becomes smaller, learned index methods can extract more general patterns in the data as some noisy keys are discarded. So some adjacent learned segments having similar slopes can be merged, and the generalization ability improves.
Note that PGM adopts an optimal piece-wise segmentation learning algorithm, and thus it is more stable than FITing-Tree that adopts a greedy learning algorithm.

\paragraph{The $\alpha$ Adjustment}.
To gain more insights about the proposed sampling technique, here, we vary the tunable $\alpha$ of learned index methods and check their the smallest ``safe'' sampled data size $n_{safe}$ that maintains a non-degraded performance. 

\begin{figure}[htbp]
	\centering
	\subfigure{
		\begin{minipage}[]{0.45\textwidth}
			\centering
			\includegraphics[width=\textwidth]{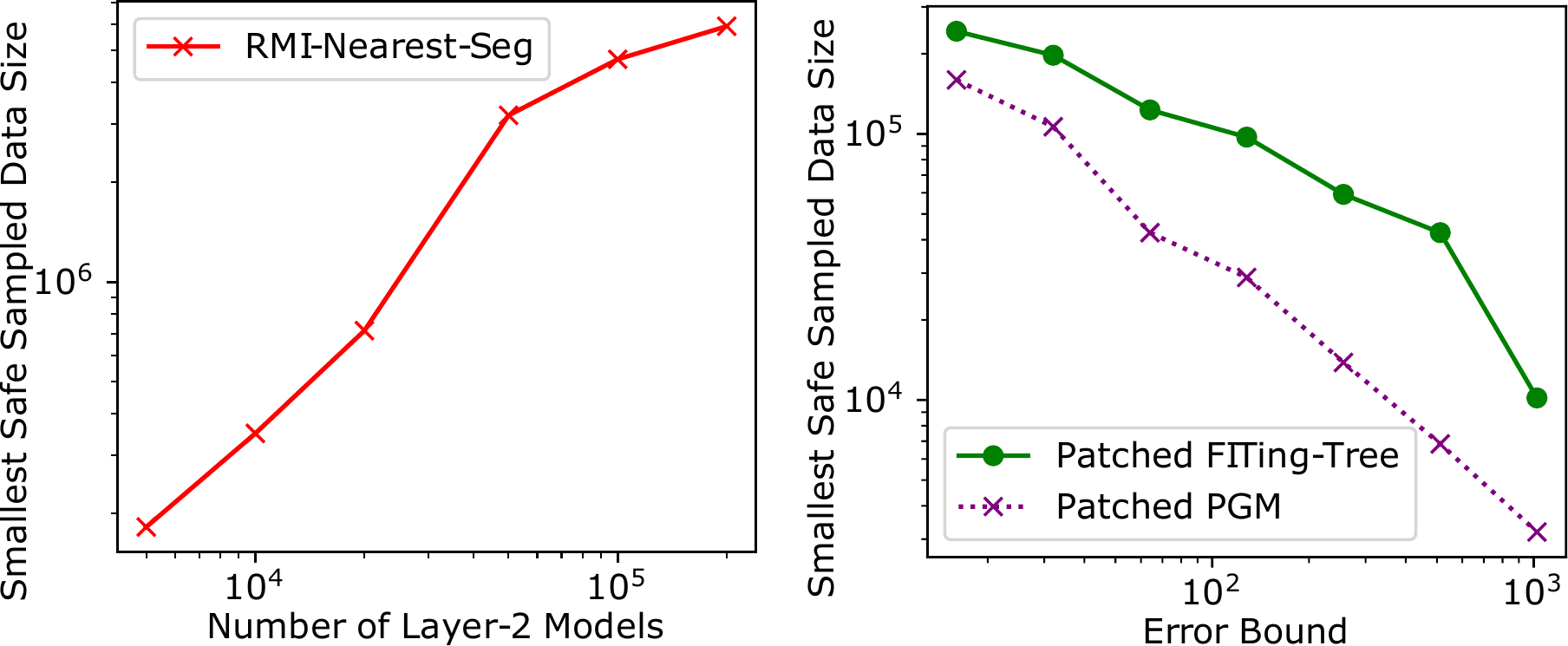}
		\end{minipage}%
	}
	\vspace{-0.2in}
	\caption{The smallest ``safe'' sampled data size $n_{safe}$ to maintain a non-degraded performance, by varying $\alpha$ of different learned index methods on IoT dataset.}
	\label{fig:sample-alpha-safe-iot}
\end{figure}

\begin{figure*}[htbp]
	\centering
	\subfigure{
		\begin{minipage}{0.16\textwidth}
			\centering
			\includegraphics[width=\linewidth]{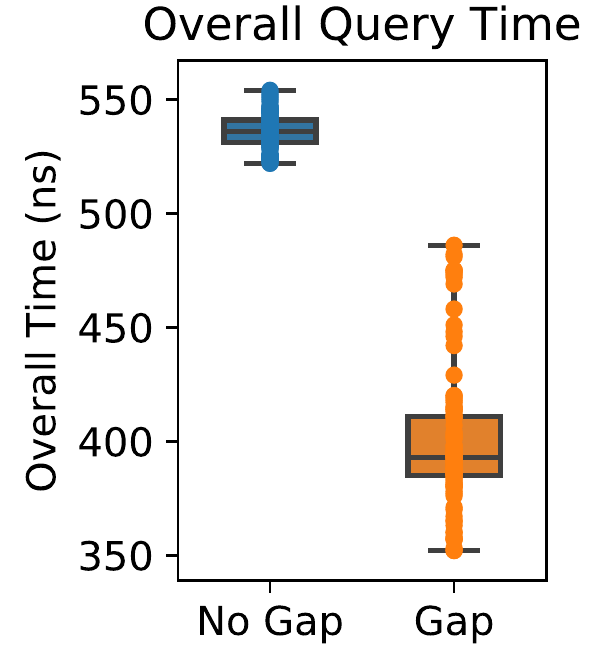}
			\vspace{-0.1in}
		\end{minipage}
	} 
	\subfigure{
	\begin{minipage}{0.16\textwidth}
		\centering
		\includegraphics[width=\linewidth]{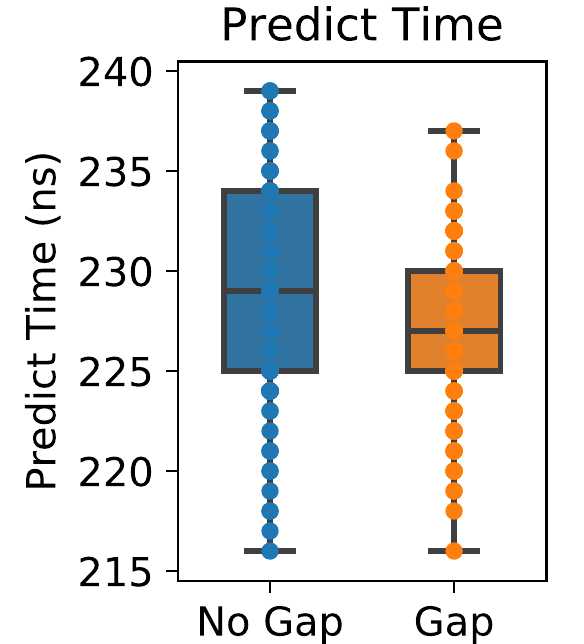}
		\vspace{-0.1in}
	\end{minipage}
} 
	\subfigure{
	\begin{minipage}{0.16\textwidth}
		\centering
		\includegraphics[width=\linewidth]{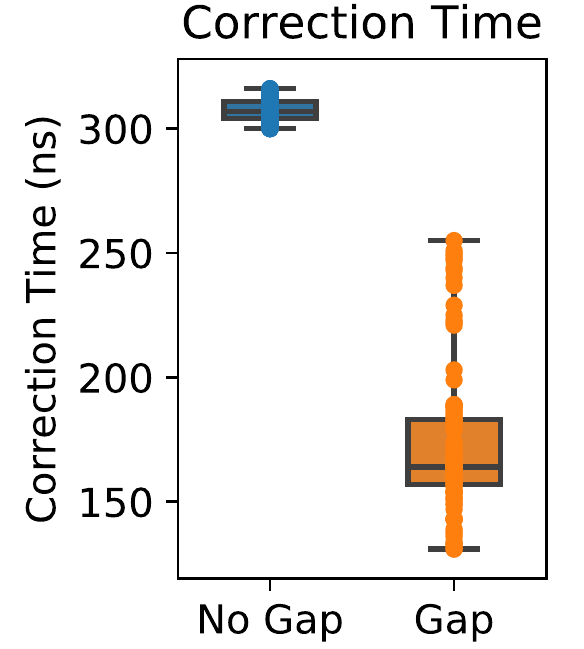}
		\vspace{-0.1in}
	\end{minipage}
} 
	\subfigure{
	\begin{minipage}{0.158\textwidth}
		\centering
		\includegraphics[width=\linewidth]{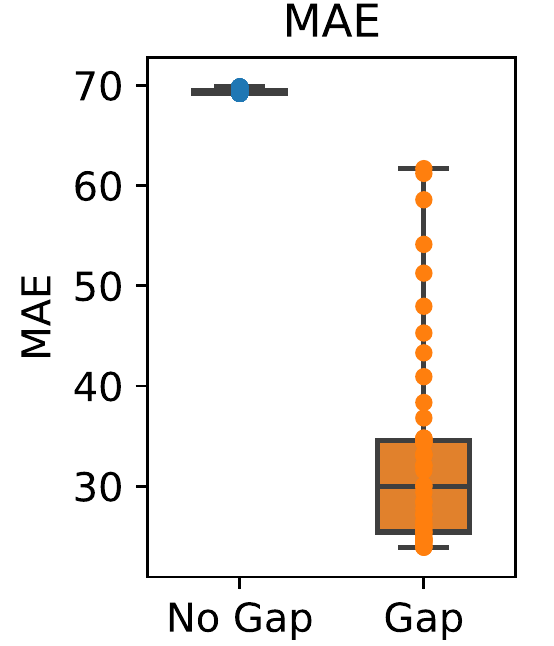}
		\vspace{-0.1in}
	\end{minipage}
} 
	\subfigure{
	\begin{minipage}{0.16\textwidth}
		\centering
		\includegraphics[width=\linewidth]{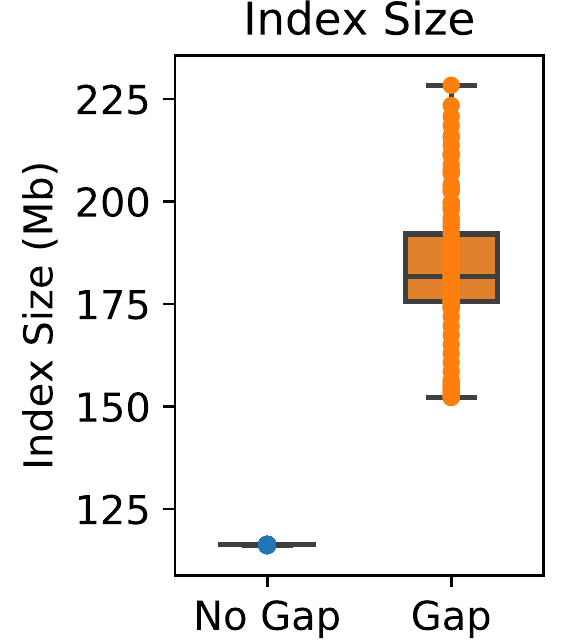}
		\vspace{-0.1in}
	\end{minipage}
} 
	\vspace{-0.23in}
	\caption{Boxplots of performance for PGM index with gap insertion on IoT dataset.}
	\vspace{-0.15in}
	\label{fig:gap-static-overall}
\end{figure*}

The results are shown in Figure \ref{fig:sample-alpha-safe-iot}. 
Recall that in Section \ref{sec:sample}, we theoretically show that it suffices to learn the index on a small sample with size as small as $O(\alpha ^2 log ^2 \mathcal{E})$ in Theorem 1. 
Applying simple log transformation, we can get $log(O(\alpha ^2 log ^2 \mathcal{E}))= O(4loglog\mathcal{E} \times log(\alpha)) = O(log(\alpha))$. 
That is, the $log(n_{safe})$ for these methods are asymptotically linear to $log(\alpha)$.
In Figure \ref{fig:sample-alpha-safe-iot}, the x-axis and y-axis represent $\alpha$ and $n_{safe}$ with log transformation respectively, and the linear trends of those plots match the above theoretical analysis (again, note that the number of layer-2 models in RMI is proportional to $\alpha$, while the error bound of Fitting-Tree and PGM are inversely proportional to $\alpha$).
With a smaller $\alpha$, we can achieve non-degraded performance with fewer samples, while a larger $\alpha$ requires us to draw more samples since more details about the data are needed to learn an index in finer granularity. 

\subsection{Learned Indexes with Gap Insertion}
\label{exp:gap}
As discussed in Section \ref{sec:gap}, we propose to learn precise index by adjusting the distribution of keys and positions, at the cost of gap insertion and index re-training. 
The cost can be further reduced by combining the sampling technique.
In this section, we empirically examine the effectiveness and efficiency of the proposed gap insertion technique.

\subsubsection{Static Scenarios}
We conduct experiments on all the adopted datasets and the three learned index methods, by varying the gap insertion rate $\rho$ from 0.5 to 0.001 and the sample rate $s$ from 1 to 0.005. Due to similar results and space limitation, we only report the results for PGM and IoT dataset here. 
The overall query time and other detailed performance numbers are summarized as boxplots in Figure \ref{fig:gap-static-overall}, where each box indicates the middle 50\% experimental points and the ``No Gap'' boxes represent the baseline index $M_{\data}$ learned without gap insertion. 

Clearly, indexes learned on the gap-inserted datasets gain significantly smaller overall query times comparing with the baseline (e.g., the speedup is up to 1.59x), which verifies the effectiveness of the gap insertion technique. Note that the overall query time is the overall performance including the advantage of improved MAE and the disadvantage of additional index size. To further analyze the detailed performance, we further break down the overall query times into the prediction time and corrections time, and also plot MAE and index size in Figure \ref{fig:gap-static-overall}. 
Comparing with the baseline without gap insertion, we can see that PGM with gap insertion achieves slightly better prediction time and much better correction time. 
To explain the improvements, we can check the MAE results, which show a significant improvement. However, the index size, including the size of introduced gaps and linking arrays, becomes larger and reduce the benefits brought by the MAE improvements. So in total, the correction time shows an averaged 2x improvement, and the overall query time shows an averaged 1.4x speedup, which is less than the improvement in terms of MAE. 
These results and analysis show that our gap insertion technique can learn preciser indexes, and thus improve overall performance.

\subsubsection{Effect of Gap Insertion and Sampling}
\label{exp:effect_of_gap}
Above we analyze the overall performance of the gap insertion and sampling techniques, here we discuss the performance of various cases with specific $s$ and $\rho$, as illustrated in Figure \ref{fig:gap-sample-pgm-iot-time}. 

\begin{figure}[h]
	\centering
	\subfigure{
		\begin{minipage}[]{0.47\textwidth}
			\centering
			\includegraphics[width=1.0\linewidth]{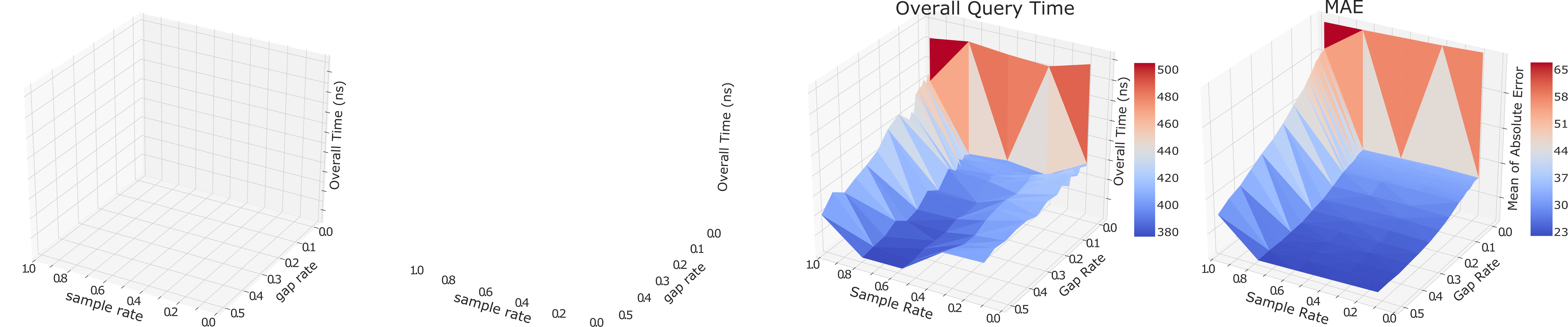}
		\end{minipage}%
	}
	\vspace{-0.2in}
	\caption{The performance of PGM with varied $s$ and $\rho$.}
	\label{fig:gap-sample-pgm-iot-time}
\end{figure}

Overall speaking, comparing with the results without gap insertion and sampling ($s=1, \rho=0$, \ie, the left upper concern in each subfigure), large gap rates and moderate sample rates achieve much smaller MAE and significant query efficiency improvements, verifying the effectiveness of proposed sampling and gap insertion techniques again.
Specifically, a larger $\rho$ allows us to transform the data distribution and enhance the patterns via the result-driven gap insertion. 
For the sampling technique, we can observe that it behaves a bit different from the results in Section \ref{exp:sample}, where both reasonable MAE and query time can be maintained as $s$ decreases from 1.0 to 0.1. Here the reasonable MAE will be maintained but overall query time will decrease firstly and then increase slightly.
This is because that when combining the sampling and gap insertion, with a very small $s$, we need to put more un-sampled keys into the linking arrays, resulting in an increased total query time.

\subsubsection{Dynamic Performance}
As mentioned in Section \ref{sec:handle-dynamic}, our proposed techniques can be easily extended to handle the dynamic scenarios. 
To show how the performance varies with different dynamic scenarios, as an example here, we evaluate PGM with dynamic linking arrays on both read-heavy and write-heavy workloads. 
Specifically, we randomly split the IoT dataset into $\data_{init}$ and $\data_{init}^-$ with write proportion $w$, \ie, $\data = \data_{init} + \data_{init}^{-}, ~~|D_{init}^{-}| = w \cdot |D|$. We choose $w=0.3$ and $w=0.7$ for read-heavy and write-heavy workloads respectively, and randomly split  $\data_{init}^-$ into $B$ equal-sized data batches $\data_{1}^{-}, \dots, D_{B}^{-}$.
We initially learn the index on $\data_{init}$ and then insert $\data_{init}^{-}$ in batches. After inserting the $b$-th batch, we evaluate its MAE, prediction time, correction time and overall query time by randomly querying the data we have seen so far.

\vspace{-0.1in}
\begin{figure}[h]
	\centering
	\subfigure{
		\begin{minipage}[]{0.48\textwidth}
			\centering
			\includegraphics[width=\textwidth]{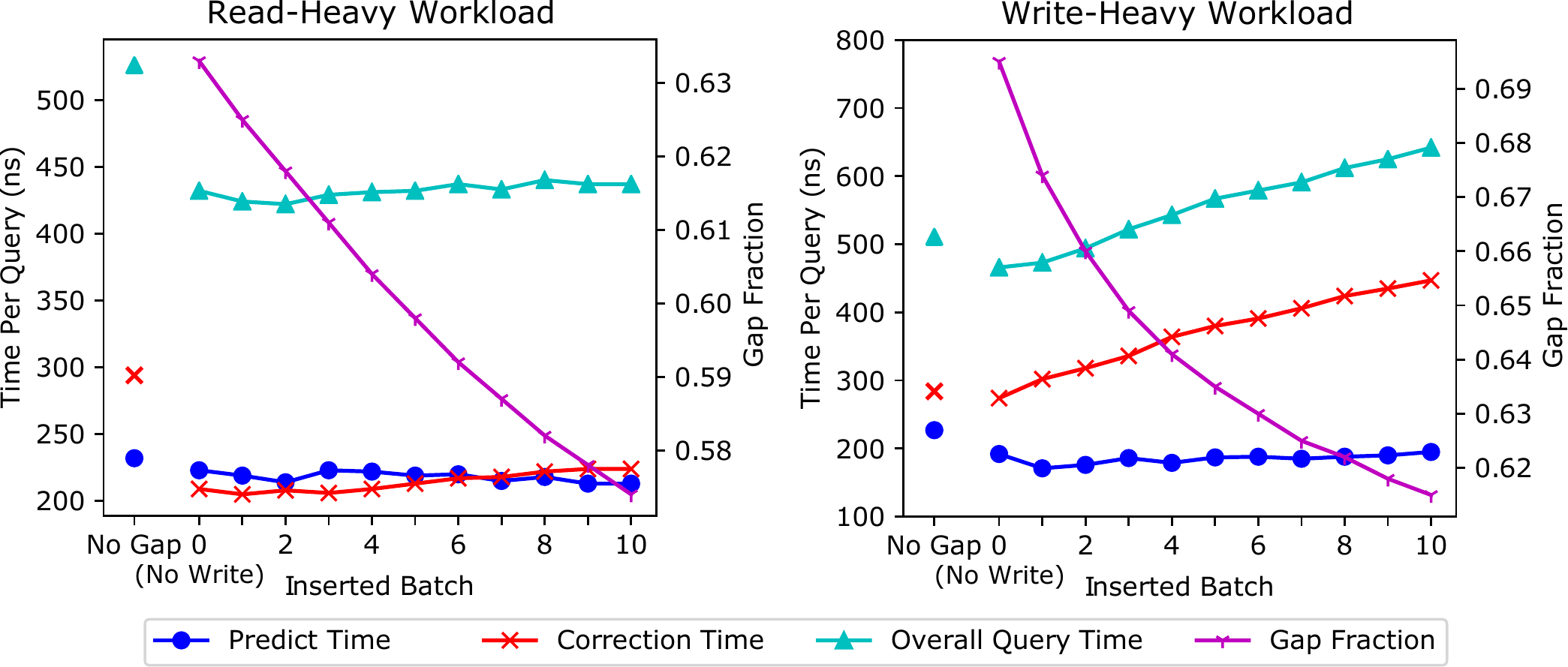}
		\end{minipage}%
	}
	\vspace{-0.15in}
	\caption{The performance of PGM with linking arrays on dynamic scenarios.}
	\vspace{-0.1in}
	\label{fig:dynamic-res}
\end{figure}

The left subfigure in Figure \ref{fig:dynamic-res} shows the querying-related time and gap fraction after each batch insertion for the read-heavy workload. As more new data are inserted into the reserved unoccupied positions, the fraction of available gaps decreases as expected. For querying-related performance, the prediction time remains the same, while the correction and overall query time slightly increase as some new data are placed into the linking arrays. 
Note that during the insertion procedure, the index with dynamic linking arrays achieves about 1.364x faster correction speed and 1.227x faster overall lookup speed than the index without gap insertion technique (the leftmost points in the left subfigure of Figure \ref{fig:dynamic-res}, and this baseline can access to all the data $\data$), showing the great potentials of the gap-based index to handle dynamic scenarios.
In the right subfigure, we also plot the results for write-heavy workload, which shows similar trends as read-heavy workload. The only difference is that the correction and overall query time increases a bit faster than the ones in ready-heavy workload, as we have more new data to insert in the write-heavy workload.
All these results confirm that the proposed method can effectively leverage the reserved gaps and maintain comparable query performance for dynamic scenarios.

\section{Related Work}
\label{sec:related}
\paragraph{Learned Indexes}. 
There have been many well-studied traditional indexes, such as tree-based \cite{art, csb, o1996log, graefe, kim2010fast, athanassoulis2014bf, bayer1977prefix}, hash-based \cite{Pagh2004cuckoo, wang2017survey} and bitmap-based \cite{bitmap1, bitmap2, Manos2016upbit, bitmap5, bitmap_perf}.
Recently, the learned indexes gain increasing interest, which learn and utilize the hidden distribution of data to be indexed.
Recursive Model Index (RMI) \cite{kraska2018case} firstly introduces the idea of predicting positions of keys with machine learning models. 
FITing-Tree \cite{galakatos2019fiting} greedily learns piece-wise linear segments with a pre-defined bounded error, and PGM \cite{Ferragina:2020pgm} further improves FITing-Tree by finding the optimal number of learned segments given an error bound. 
ALEX \cite{ding2019alex} proposes an adaptive RMI with workload-specific optimization, achieving high performance on dynamic workloads. 
RadixSpline \cite{kipf2020radixspline} gains competitive performance to RMI with a radix structure and a single-pass training.
In addition to single-dimensional indexes, existing methods also explore the multi-dimensional scenarios, such as Flood \cite{nathan2020flood}, Tsunami \cite{ding2020tsunami}, NEIST \cite{wu2019neist}  and LISA \cite{li2020lisa}.

Our work is kind of complimentary to above existing learned index methods. We propose an MDL-based learning framework to quantify the learning objective and propose two pluggable techniques that can be incorporated by existing methods to further enhances their performance as shown in experiments. 

\paragraph{Sampling in Index}.
Sampling has also been explored in partial index \cite{stonebraker1989case, seshadri1995generalized}, tail index \cite{galakatos2017revisiting}, cost estimation \cite{ding2019alex, crotty2016case} and index layout estimation \cite{lang2001modeling, nathan2020flood}.
Different from sampling a user-interested subset in partial indexes \cite{stonebraker1989case} or a rare subset in tail index \cite{galakatos2017revisiting}, we adopt a uniform sampling in our experiments.
ALEX \cite{ding2019alex} samples every $n$-th key of data to predict the cost of data node for bulk loading.
Flood \cite{nathan2020flood} trains RMI on each dimension and samples data to estimate how often certain dimensions are used. 
Different from ALEX and Flood, we propose the sampling technique for learning acceleration, and we also provide some theoretical analysis that is further confirmed by experimental results. 

\paragraph{Gapped Structure}.
To support dynamic operations, several gapped data structures have been studied to reserve gaps between elements, including Packed Memory Array (PMA) \cite{bender2007adaptive}, Packed-Memory Quadtree \cite{toss2018packed}, 
B+Tree that reserves continuous gaps at the end of data arrays, 
and ALEX \cite{ding2019alex} that adopts a gapped array with key shifts and model-based insertion.
PMA and B+Tree reserve gaps that are independent with different data distributions, while 
our method and ALEX reserve data-dependent gap using learned models. Further, different from ALEX, we use linking arrays to simplify the dynamic operations and make the key-position distribution of gap-inserted data suitable for possible inserted data.

\paragraph{Machine Learning based Database}.
We compare several learned indexes from the perspective of minimum description length \cite{grunwald2007minimum, grunwald2019minimum}, which falls in the category of machine learning based database. 
Recently, there are many works facilitating database components with machine learning, such as query optimization \cite{krishnan2018learning, marcus2019neo, marcus2018towards, ortiz2018learning, kipf2018learned, Dutt2019selectivity}, workload forecasting \cite{ma2018query}, memory prefetchers \cite{milad2018learning}, and selectivity estimation \cite{yang2019selectivity}. 
\section{Conclusion}
Learned index gains promising performance by learning and utilizing the hidden distribution of the data to be indexed.
To facilitate the learned index from the view of machine learning, we propose a minimum description length based framework that can formally quantify \textit{index learning objective} and help to design suitable learning indexes for different scenarios. 
Besides, we study two general and pluggable techniques, \ie, the sampling technique to enhance \textit{learning efficiency} with theoretical guidance, and the result-driven gap insertion technique to enhance \textit{learning effectiveness} in terms of index preciseness and generalization ability.
Extensive experiments demonstrate the efficiency and effectiveness of the proposed framework and the two pluggable techniques, which boost existing learned index methods by up to 78x construction speedup meanwhile maintaining non-degraded performance, and up to 1.59x query speedup on both static and dynamic indexing scenarios.
With this paper, we hope to provide a deeper understanding of current learned index methods from the perspective of machine learning, and promote more explorations of learned index from both the perspective of machine learning and database.

\newpage
\bibliographystyle{ACM-Reference-Format}
\bibliography{ref}


\begin{thebibliography}{53}


\ifx \showCODEN    \undefined \def \showCODEN     #1{\unskip}     \fi
\ifx \showDOI      \undefined \def \showDOI       #1{#1}\fi
\ifx \showISBNx    \undefined \def \showISBNx     #1{\unskip}     \fi
\ifx \showISBNxiii \undefined \def \showISBNxiii  #1{\unskip}     \fi
\ifx \showISSN     \undefined \def \showISSN      #1{\unskip}     \fi
\ifx \showLCCN     \undefined \def \showLCCN      #1{\unskip}     \fi
\ifx \shownote     \undefined \def \shownote      #1{#1}          \fi
\ifx \showarticletitle \undefined \def \showarticletitle #1{#1}   \fi
\ifx \showURL      \undefined \def \showURL       {\relax}        \fi
\providecommand\bibfield[2]{#2}
\providecommand\bibinfo[2]{#2}
\providecommand\natexlab[1]{#1}
\providecommand\showeprint[2][]{arXiv:#2}

\bibitem[\protect\citeauthoryear{Abel}{Abel}{1984}]%
        {abel1984b+}
\bibfield{author}{\bibinfo{person}{David~J Abel}.}
  \bibinfo{year}{1984}\natexlab{}.
\newblock \showarticletitle{A {B}+-tree structure for large quadtrees}.
\newblock \bibinfo{journal}{\emph{Computer Vision, Graphics, and Image
  Processing}} \bibinfo{volume}{27}, \bibinfo{number}{1}
  (\bibinfo{year}{1984}), \bibinfo{pages}{19--31}.
\newblock


\bibitem[\protect\citeauthoryear{Athanassoulis et~al\mbox{.}}{Athanassoulis
  et~al\mbox{.}}{2016}]%
        {Manos2016upbit}
\bibfield{author}{\bibinfo{person}{Manos Athanassoulis} {et~al\mbox{.}}}
  \bibinfo{year}{2016}\natexlab{}.
\newblock \showarticletitle{{UpBit: Scalable In-Memory Updatable Bitmap
  Indexing}}. In \bibinfo{booktitle}{\emph{SIGMOD}}.
  \bibinfo{pages}{1319--1332}.
\newblock


\bibitem[\protect\citeauthoryear{Athanassoulis and Ailamaki}{Athanassoulis and
  Ailamaki}{2014}]%
        {athanassoulis2014bf}
\bibfield{author}{\bibinfo{person}{Manos Athanassoulis} {and}
  \bibinfo{person}{Anastasia Ailamaki}.} \bibinfo{year}{2014}\natexlab{}.
\newblock \showarticletitle{{BF-tree: Approximate Tree Indexing}}. In
  \bibinfo{booktitle}{\emph{{PVLDB}}}. \bibinfo{pages}{1881--1892}.
\newblock


\bibitem[\protect\citeauthoryear{Bayer and Unterauer}{Bayer and
  Unterauer}{1977}]%
        {bayer1977prefix}
\bibfield{author}{\bibinfo{person}{Rudolf Bayer} {and} \bibinfo{person}{Karl
  Unterauer}.} \bibinfo{year}{1977}\natexlab{}.
\newblock \showarticletitle{Prefix {B}-trees}.
\newblock \bibinfo{journal}{\emph{TODS}} \bibinfo{volume}{2},
  \bibinfo{number}{1} (\bibinfo{year}{1977}), \bibinfo{pages}{11--26}.
\newblock


\bibitem[\protect\citeauthoryear{Bender and Hu}{Bender and Hu}{2007}]%
        {bender2007adaptive}
\bibfield{author}{\bibinfo{person}{Michael~A Bender} {and}
  \bibinfo{person}{Haodong Hu}.} \bibinfo{year}{2007}\natexlab{}.
\newblock \showarticletitle{An adaptive packed-memory array}.
\newblock \bibinfo{journal}{\emph{TODS}} \bibinfo{volume}{32},
  \bibinfo{number}{4} (\bibinfo{year}{2007}), \bibinfo{pages}{26}.
\newblock


\bibitem[\protect\citeauthoryear{Bertsekas}{Bertsekas}{1997}]%
        {bertsekas1997nonlinear}
\bibfield{author}{\bibinfo{person}{Dimitri~P Bertsekas}.}
  \bibinfo{year}{1997}\natexlab{}.
\newblock \showarticletitle{Nonlinear programming}.
\newblock \bibinfo{journal}{\emph{Journal of the Operational Research Society}}
  \bibinfo{volume}{48}, \bibinfo{number}{3} (\bibinfo{year}{1997}),
  \bibinfo{pages}{334--334}.
\newblock


\bibitem[\protect\citeauthoryear{Bingmann}{Bingmann}{2013}]%
        {stxbtree}
\bibfield{author}{\bibinfo{person}{Timo Bingmann}.}
  \bibinfo{year}{2013}\natexlab{}.
\newblock \bibinfo{title}{STX B+ Tree}.
\newblock \bibinfo{howpublished}{\url{https://panthema.net/2007/stx-btree/}}.
\newblock


\bibitem[\protect\citeauthoryear{Bishop}{Bishop}{2006}]%
        {bishop2006pattern}
\bibfield{author}{\bibinfo{person}{Christopher~M Bishop}.}
  \bibinfo{year}{2006}\natexlab{}.
\newblock \bibinfo{booktitle}{\emph{Pattern recognition and machine learning}}.
\newblock \bibinfo{publisher}{springer}.
\newblock


\bibitem[\protect\citeauthoryear{Chan and Ioannidis}{Chan and
  Ioannidis}{1998}]%
        {bitmap1}
\bibfield{author}{\bibinfo{person}{Chee-Yong Chan} {and}
  \bibinfo{person}{Yannis~E. Ioannidis}.} \bibinfo{year}{1998}\natexlab{}.
\newblock \showarticletitle{{Bitmap Index Design and Evaluation}}. In
  \bibinfo{booktitle}{\emph{{SIGMOD}}}. \bibinfo{pages}{355--366}.
\newblock


\bibitem[\protect\citeauthoryear{Chang, Chang, Wu, and Wu}{Chang
  et~al\mbox{.}}{2000}]%
        {chang2000b}
\bibfield{author}{\bibinfo{person}{Yun-Chih Chang}, \bibinfo{person}{Yao-Wen
  Chang}, \bibinfo{person}{Guang-Ming Wu}, {and} \bibinfo{person}{Shu-Wei Wu}.}
  \bibinfo{year}{2000}\natexlab{}.
\newblock \showarticletitle{B*-Trees: a new representation for non-slicing
  floorplans}. In \bibinfo{booktitle}{\emph{DAC}}. \bibinfo{pages}{458--463}.
\newblock


\bibitem[\protect\citeauthoryear{Chung, Çaglar G{\"u}lçehre, Cho, and
  Bengio}{Chung et~al\mbox{.}}{2014}]%
        {Chung2014EmpiricalEO}
\bibfield{author}{\bibinfo{person}{Junyoung Chung}, \bibinfo{person}{Çaglar
  G{\"u}lçehre}, \bibinfo{person}{Kyunghyun Cho}, {and}
  \bibinfo{person}{Yoshua Bengio}.} \bibinfo{year}{2014}\natexlab{}.
\newblock \showarticletitle{Empirical Evaluation of Gated Recurrent Neural
  Networks on Sequence Modeling}.
\newblock \bibinfo{journal}{\emph{ArXiv}}  \bibinfo{volume}{abs/1412.3555}
  (\bibinfo{year}{2014}).
\newblock


\bibitem[\protect\citeauthoryear{Crotty, Galakatos, Zgraggen, Binnig, and
  Kraska}{Crotty et~al\mbox{.}}{2016}]%
        {crotty2016case}
\bibfield{author}{\bibinfo{person}{Andrew Crotty}, \bibinfo{person}{Alex
  Galakatos}, \bibinfo{person}{Emanuel Zgraggen}, \bibinfo{person}{Carsten
  Binnig}, {and} \bibinfo{person}{Tim Kraska}.}
  \bibinfo{year}{2016}\natexlab{}.
\newblock \showarticletitle{The case for interactive data exploration
  accelerators (IDEAs)}. In \bibinfo{booktitle}{\emph{HILDA}}.
  \bibinfo{pages}{1--6}.
\newblock


\bibitem[\protect\citeauthoryear{Ding, Minhas, Zhang, Li, Wang, Chandramouli,
  Gehrke, Kossmann, and Lomet}{Ding et~al\mbox{.}}{2020a}]%
        {ding2019alex}
\bibfield{author}{\bibinfo{person}{Jialin Ding}, \bibinfo{person}{Umar~Farooq
  Minhas}, \bibinfo{person}{Hantian Zhang}, \bibinfo{person}{Yinan Li},
  \bibinfo{person}{Chi Wang}, \bibinfo{person}{Badrish Chandramouli},
  \bibinfo{person}{Johannes Gehrke}, \bibinfo{person}{Donald Kossmann}, {and}
  \bibinfo{person}{David~B. Lomet}.} \bibinfo{year}{2020}\natexlab{a}.
\newblock \showarticletitle{ALEX: An Updatable Adaptive Learned Index}.
\newblock \bibinfo{journal}{\emph{SIGMOD}}, \bibinfo{pages}{969–984}.
\newblock


\bibitem[\protect\citeauthoryear{Ding, Nathan, Alizadeh, and Kraska}{Ding
  et~al\mbox{.}}{2020b}]%
        {ding2020tsunami}
\bibfield{author}{\bibinfo{person}{Jialin Ding}, \bibinfo{person}{Vikram
  Nathan}, \bibinfo{person}{Mohammad Alizadeh}, {and} \bibinfo{person}{Tim
  Kraska}.} \bibinfo{year}{2020}\natexlab{b}.
\newblock \showarticletitle{Tsunami: A Learned Multi-dimensional Index for
  Correlated Data and Skewed Workloads}.
\newblock \bibinfo{journal}{\emph{arXiv preprint arXiv:2006.13282}}
  (\bibinfo{year}{2020}).
\newblock


\bibitem[\protect\citeauthoryear{Dutt, Wang, Nazi, Kandula, Narasayya, and
  Chaudhuri}{Dutt et~al\mbox{.}}{2019}]%
        {Dutt2019selectivity}
\bibfield{author}{\bibinfo{person}{Anshuman Dutt}, \bibinfo{person}{Chi Wang},
  \bibinfo{person}{Azade Nazi}, \bibinfo{person}{Srikanth Kandula},
  \bibinfo{person}{Vivek Narasayya}, {and} \bibinfo{person}{Surajit
  Chaudhuri}.} \bibinfo{year}{2019}\natexlab{}.
\newblock \showarticletitle{Selectivity Estimation for Range Predicates Using
  Lightweight Models}.
\newblock \bibinfo{journal}{\emph{PVLDB}} \bibinfo{volume}{12},
  \bibinfo{number}{9}, \bibinfo{pages}{1044–1057}.
\newblock
\showISSN{2150-8097}


\bibitem[\protect\citeauthoryear{Ferragina and Vinciguerra}{Ferragina and
  Vinciguerra}{2020}]%
        {Ferragina:2020pgm}
\bibfield{author}{\bibinfo{person}{Paolo Ferragina} {and}
  \bibinfo{person}{Giorgio Vinciguerra}.} \bibinfo{year}{2020}\natexlab{}.
\newblock \showarticletitle{The {PGM-index}: a fully-dynamic compressed learned
  index with provable worst-case bounds}.
\newblock \bibinfo{journal}{\emph{{PVLDB}}} \bibinfo{volume}{13},
  \bibinfo{number}{8}, \bibinfo{pages}{1162--1175}.
\newblock
\showISSN{2150-8097}


\bibitem[\protect\citeauthoryear{Feurer and Hutter}{Feurer and Hutter}{2019}]%
        {feurer2019hyperparameter}
\bibfield{author}{\bibinfo{person}{Matthias Feurer} {and}
  \bibinfo{person}{Frank Hutter}.} \bibinfo{year}{2019}\natexlab{}.
\newblock \showarticletitle{Hyperparameter optimization}.
\newblock In \bibinfo{booktitle}{\emph{Automated Machine Learning}}.
  \bibinfo{publisher}{Springer, Cham}, \bibinfo{pages}{3--33}.
\newblock


\bibitem[\protect\citeauthoryear{Galakatos, Crotty, Zgraggen, Binnig, and
  Kraska}{Galakatos et~al\mbox{.}}{2017}]%
        {galakatos2017revisiting}
\bibfield{author}{\bibinfo{person}{Alex Galakatos}, \bibinfo{person}{Andrew
  Crotty}, \bibinfo{person}{Emanuel Zgraggen}, \bibinfo{person}{Carsten
  Binnig}, {and} \bibinfo{person}{Tim Kraska}.}
  \bibinfo{year}{2017}\natexlab{}.
\newblock \showarticletitle{Revisiting reuse for approximate query processing}.
\newblock \bibinfo{journal}{\emph{PVLDB}} \bibinfo{volume}{10},
  \bibinfo{number}{10}, \bibinfo{pages}{1142--1153}.
\newblock


\bibitem[\protect\citeauthoryear{Galakatos, Markovitch, Binnig, Fonseca, and
  Kraska}{Galakatos et~al\mbox{.}}{2019}]%
        {galakatos2019fiting}
\bibfield{author}{\bibinfo{person}{Alex Galakatos}, \bibinfo{person}{Michael
  Markovitch}, \bibinfo{person}{Carsten Binnig}, \bibinfo{person}{Rodrigo
  Fonseca}, {and} \bibinfo{person}{Tim Kraska}.}
  \bibinfo{year}{2019}\natexlab{}.
\newblock \showarticletitle{{FIT}ing-{T}ree: A Data-aware Index Structure}. In
  \bibinfo{booktitle}{\emph{SIGMOD}}. \bibinfo{pages}{1189--1206}.
\newblock


\bibitem[\protect\citeauthoryear{Graefe and Larson}{Graefe and Larson}{2001}]%
        {graefe}
\bibfield{author}{\bibinfo{person}{Goetz Graefe} {and}
  \bibinfo{person}{Per{-}{\AA}ke Larson}.} \bibinfo{year}{2001}\natexlab{}.
\newblock \showarticletitle{{B-Tree Indexes and {CPU} Caches}}. In
  \bibinfo{booktitle}{\emph{{ICDE}}}. \bibinfo{pages}{349--358}.
\newblock


\bibitem[\protect\citeauthoryear{Gr{\"u}nwald and Roos}{Gr{\"u}nwald and
  Roos}{2019}]%
        {grunwald2019minimum}
\bibfield{author}{\bibinfo{person}{Peter Gr{\"u}nwald} {and}
  \bibinfo{person}{Teemu Roos}.} \bibinfo{year}{2019}\natexlab{}.
\newblock \showarticletitle{Minimum description length revisited}.
\newblock \bibinfo{journal}{\emph{arXiv preprint arXiv:1908.08484}}
  (\bibinfo{year}{2019}).
\newblock


\bibitem[\protect\citeauthoryear{Gr{\"u}nwald and Grunwald}{Gr{\"u}nwald and
  Grunwald}{2007}]%
        {grunwald2007minimum}
\bibfield{author}{\bibinfo{person}{Peter~D Gr{\"u}nwald} {and}
  \bibinfo{person}{Abhijit Grunwald}.} \bibinfo{year}{2007}\natexlab{}.
\newblock \bibinfo{booktitle}{\emph{The minimum description length principle}}.
\newblock \bibinfo{publisher}{MIT press}.
\newblock


\bibitem[\protect\citeauthoryear{Hashemi, Swersky, Smith, Ayers, Litz, Chang,
  Kozyrakis, and Ranganathan}{Hashemi et~al\mbox{.}}{2018}]%
        {milad2018learning}
\bibfield{author}{\bibinfo{person}{Milad Hashemi}, \bibinfo{person}{Kevin
  Swersky}, \bibinfo{person}{Jamie~A. Smith}, \bibinfo{person}{Grant Ayers},
  \bibinfo{person}{Heiner Litz}, \bibinfo{person}{Jichuan Chang},
  \bibinfo{person}{Christos Kozyrakis}, {and} \bibinfo{person}{Parthasarathy
  Ranganathan}.} \bibinfo{year}{2018}\natexlab{}.
\newblock \showarticletitle{Learning Memory Access Patterns}. In
  \bibinfo{booktitle}{\emph{ICML}}.
\newblock


\bibitem[\protect\citeauthoryear{Jagadish, Ooi, Tan, Yu, and Zhang}{Jagadish
  et~al\mbox{.}}{2005}]%
        {jagadish2005idistance}
\bibfield{author}{\bibinfo{person}{Hosagrahar~V Jagadish},
  \bibinfo{person}{Beng~Chin Ooi}, \bibinfo{person}{Kian-Lee Tan},
  \bibinfo{person}{Cui Yu}, {and} \bibinfo{person}{Rui Zhang}.}
  \bibinfo{year}{2005}\natexlab{}.
\newblock \showarticletitle{iDistance: An adaptive {B}+-tree based indexing
  method for nearest neighbor search}.
\newblock \bibinfo{journal}{\emph{TODS}} \bibinfo{volume}{30},
  \bibinfo{number}{2} (\bibinfo{year}{2005}), \bibinfo{pages}{364--397}.
\newblock


\bibitem[\protect\citeauthoryear{Johnson}{Johnson}{1999}]%
        {bitmap_perf}
\bibfield{author}{\bibinfo{person}{Theodore Johnson}.}
  \bibinfo{year}{1999}\natexlab{}.
\newblock \showarticletitle{{Performance Measurements of Compressed Bitmap
  Indices}}. In \bibinfo{booktitle}{\emph{{PVLDB}}}. \bibinfo{pages}{278--289}.
\newblock


\bibitem[\protect\citeauthoryear{Kim, Chhugani, Satish, Sedlar, Nguyen,
  Kaldewey, Lee, Brandt, and Dubey}{Kim et~al\mbox{.}}{2010}]%
        {kim2010fast}
\bibfield{author}{\bibinfo{person}{Changkyu Kim}, \bibinfo{person}{Jatin
  Chhugani}, \bibinfo{person}{Nadathur Satish}, \bibinfo{person}{Eric Sedlar},
  \bibinfo{person}{Anthony~D Nguyen}, \bibinfo{person}{Tim Kaldewey},
  \bibinfo{person}{Victor~W Lee}, \bibinfo{person}{Scott~A Brandt}, {and}
  \bibinfo{person}{Pradeep Dubey}.} \bibinfo{year}{2010}\natexlab{}.
\newblock \showarticletitle{{FAST}: fast architecture sensitive tree search on
  modern {CPU}s and {GPU}s}. In \bibinfo{booktitle}{\emph{SIGMOD}}.
  \bibinfo{pages}{339--350}.
\newblock


\bibitem[\protect\citeauthoryear{Kipf, Kipf, Radke, Leis, Boncz, and
  Kemper}{Kipf et~al\mbox{.}}{2018}]%
        {kipf2018learned}
\bibfield{author}{\bibinfo{person}{Andreas Kipf}, \bibinfo{person}{Thomas
  Kipf}, \bibinfo{person}{Bernhard Radke}, \bibinfo{person}{Viktor Leis},
  \bibinfo{person}{Peter Boncz}, {and} \bibinfo{person}{Alfons Kemper}.}
  \bibinfo{year}{2018}\natexlab{}.
\newblock \showarticletitle{Learned Cardinalities: Estimating Correlated Joins
  with Deep Learning}.
\newblock \bibinfo{journal}{\emph{arXiv preprint arXiv:1809.00677}}
  (\bibinfo{year}{2018}).
\newblock


\bibitem[\protect\citeauthoryear{Kipf, Marcus, van Renen, Stoian, Kemper,
  Kraska, and Neumann}{Kipf et~al\mbox{.}}{2020}]%
        {kipf2020radixspline}
\bibfield{author}{\bibinfo{person}{Andreas Kipf}, \bibinfo{person}{Ryan
  Marcus}, \bibinfo{person}{Alexander van Renen}, \bibinfo{person}{Mihail
  Stoian}, \bibinfo{person}{Alfons Kemper}, \bibinfo{person}{Tim Kraska}, {and}
  \bibinfo{person}{Thomas Neumann}.} \bibinfo{year}{2020}\natexlab{}.
\newblock \showarticletitle{RadixSpline: a single-pass learned index}.
\newblock \bibinfo{journal}{\emph{arXiv preprint arXiv:2004.14541}}
  (\bibinfo{year}{2020}).
\newblock


\bibitem[\protect\citeauthoryear{Kraska, Beutel, Chi, Dean, and
  Polyzotis}{Kraska et~al\mbox{.}}{2018}]%
        {kraska2018case}
\bibfield{author}{\bibinfo{person}{Tim Kraska}, \bibinfo{person}{Alex Beutel},
  \bibinfo{person}{Ed~H Chi}, \bibinfo{person}{Jeffrey Dean}, {and}
  \bibinfo{person}{Neoklis Polyzotis}.} \bibinfo{year}{2018}\natexlab{}.
\newblock \showarticletitle{The case for learned index structures}. In
  \bibinfo{booktitle}{\emph{SIGMOD}}. \bibinfo{pages}{489--504}.
\newblock


\bibitem[\protect\citeauthoryear{Krishnan, Yang, Goldberg, Hellerstein, and
  Stoica}{Krishnan et~al\mbox{.}}{2018}]%
        {krishnan2018learning}
\bibfield{author}{\bibinfo{person}{Sanjay Krishnan}, \bibinfo{person}{Zongheng
  Yang}, \bibinfo{person}{Ken Goldberg}, \bibinfo{person}{Joseph Hellerstein},
  {and} \bibinfo{person}{Ion Stoica}.} \bibinfo{year}{2018}\natexlab{}.
\newblock \showarticletitle{Learning to optimize join queries with deep
  reinforcement learning}.
\newblock \bibinfo{journal}{\emph{arXiv preprint arXiv:1808.03196}}
  (\bibinfo{year}{2018}).
\newblock


\bibitem[\protect\citeauthoryear{Lang and Singh}{Lang and Singh}{2001}]%
        {lang2001modeling}
\bibfield{author}{\bibinfo{person}{Christian~A Lang} {and}
  \bibinfo{person}{Ambuj~K Singh}.} \bibinfo{year}{2001}\natexlab{}.
\newblock \showarticletitle{Modeling high-dimensional index structures using
  sampling}. In \bibinfo{booktitle}{\emph{SIGMOD}}. \bibinfo{pages}{389--400}.
\newblock


\bibitem[\protect\citeauthoryear{Leis et~al\mbox{.}}{Leis
  et~al\mbox{.}}{2013}]%
        {art}
\bibfield{author}{\bibinfo{person}{Viktor Leis} {et~al\mbox{.}}}
  \bibinfo{year}{2013}\natexlab{}.
\newblock \showarticletitle{{The Adaptive Radix Tree: ARTful Indexing for
  Main-memory Databases}}. In \bibinfo{booktitle}{\emph{{ICDE}}}.
  \bibinfo{pages}{38--49}.
\newblock


\bibitem[\protect\citeauthoryear{Li, Lu, Zheng, Yang, and Pan}{Li
  et~al\mbox{.}}{2020}]%
        {li2020lisa}
\bibfield{author}{\bibinfo{person}{Pengfei Li}, \bibinfo{person}{Hua Lu},
  \bibinfo{person}{Qian Zheng}, \bibinfo{person}{Long Yang}, {and}
  \bibinfo{person}{Gang Pan}.} \bibinfo{year}{2020}\natexlab{}.
\newblock \showarticletitle{LISA: A Learned Index Structure for Spatial Data}.
  In \bibinfo{booktitle}{\emph{SIGMOD}}. \bibinfo{pages}{2119--2133}.
\newblock


\bibitem[\protect\citeauthoryear{Ma, Van~Aken, Hefny, Mezerhane, Pavlo, and
  Gordon}{Ma et~al\mbox{.}}{2018}]%
        {ma2018query}
\bibfield{author}{\bibinfo{person}{Lin Ma}, \bibinfo{person}{Dana Van~Aken},
  \bibinfo{person}{Ahmed Hefny}, \bibinfo{person}{Gustavo Mezerhane},
  \bibinfo{person}{Andrew Pavlo}, {and} \bibinfo{person}{Geoffrey~J Gordon}.}
  \bibinfo{year}{2018}\natexlab{}.
\newblock \showarticletitle{Query-based workload forecasting for self-driving
  database management systems}. In \bibinfo{booktitle}{\emph{SIGMOD}}.
  \bibinfo{pages}{631--645}.
\newblock


\bibitem[\protect\citeauthoryear{Marcus, Negi, Mao, Zhang, Alizadeh, Kraska,
  Papaemmanouil, and Tatbul}{Marcus et~al\mbox{.}}{2019}]%
        {marcus2019neo}
\bibfield{author}{\bibinfo{person}{Ryan Marcus}, \bibinfo{person}{Parimarjan
  Negi}, \bibinfo{person}{Hongzi Mao}, \bibinfo{person}{Chi Zhang},
  \bibinfo{person}{Mohammad Alizadeh}, \bibinfo{person}{Tim Kraska},
  \bibinfo{person}{Olga Papaemmanouil}, {and} \bibinfo{person}{Nesime Tatbul}.}
  \bibinfo{year}{2019}\natexlab{}.
\newblock \showarticletitle{Neo: A learned query optimizer}.
\newblock \bibinfo{journal}{\emph{PVLDB}}.
\newblock


\bibitem[\protect\citeauthoryear{Marcus and Papaemmanouil}{Marcus and
  Papaemmanouil}{2018}]%
        {marcus2018towards}
\bibfield{author}{\bibinfo{person}{Ryan Marcus} {and} \bibinfo{person}{Olga
  Papaemmanouil}.} \bibinfo{year}{2018}\natexlab{}.
\newblock \showarticletitle{Towards a hands-free query optimizer through deep
  learning}.
\newblock \bibinfo{journal}{\emph{arXiv preprint arXiv:1809.10212}}
  (\bibinfo{year}{2018}).
\newblock


\bibitem[\protect\citeauthoryear{Nathan, Ding, Alizadeh, and Kraska}{Nathan
  et~al\mbox{.}}{2020}]%
        {nathan2020flood}
\bibfield{author}{\bibinfo{person}{Vikram Nathan}, \bibinfo{person}{Jialin
  Ding}, \bibinfo{person}{Mohammad Alizadeh}, {and} \bibinfo{person}{Tim
  Kraska}.} \bibinfo{year}{2020}\natexlab{}.
\newblock \showarticletitle{Learning Multi-dimensional Indexes}. In
  \bibinfo{booktitle}{\emph{SIGMOD}}. \bibinfo{pages}{985--1000}.
\newblock


\bibitem[\protect\citeauthoryear{{OpenStreetMap contributors}}{{OpenStreetMap
  contributors}}{2017}]%
        {OpenStreetMap}
\bibfield{author}{\bibinfo{person}{{OpenStreetMap contributors}}.}
  \bibinfo{year}{2017}\natexlab{}.
\newblock \bibinfo{title}{{Planet dump retrieved from https://planet.osm.org
  }}.
\newblock \bibinfo{howpublished}{\url{ https://www.openstreetmap.org }}.
\newblock


\bibitem[\protect\citeauthoryear{Ortiz, Balazinska, Gehrke, and Keerthi}{Ortiz
  et~al\mbox{.}}{2018}]%
        {ortiz2018learning}
\bibfield{author}{\bibinfo{person}{Jennifer Ortiz}, \bibinfo{person}{Magdalena
  Balazinska}, \bibinfo{person}{Johannes Gehrke}, {and}
  \bibinfo{person}{S~Sathiya Keerthi}.} \bibinfo{year}{2018}\natexlab{}.
\newblock \showarticletitle{Learning state representations for query
  optimization with deep reinforcement learning}.
\newblock \bibinfo{journal}{\emph{arXiv preprint arXiv:1803.08604}}
  (\bibinfo{year}{2018}).
\newblock


\bibitem[\protect\citeauthoryear{O’Neil, Cheng, Gawlick, and
  O’Neil}{O’Neil et~al\mbox{.}}{1996}]%
        {o1996log}
\bibfield{author}{\bibinfo{person}{Patrick O’Neil}, \bibinfo{person}{Edward
  Cheng}, \bibinfo{person}{Dieter Gawlick}, {and} \bibinfo{person}{Elizabeth
  O’Neil}.} \bibinfo{year}{1996}\natexlab{}.
\newblock \showarticletitle{The log-structured merge-tree ({LSM}-tree)}.
\newblock \bibinfo{journal}{\emph{Acta Informatica}} \bibinfo{volume}{33},
  \bibinfo{number}{4} (\bibinfo{year}{1996}), \bibinfo{pages}{351--385}.
\newblock


\bibitem[\protect\citeauthoryear{Pagh and Rodler}{Pagh and Rodler}{2004}]%
        {Pagh2004cuckoo}
\bibfield{author}{\bibinfo{person}{Rasmus Pagh} {and}
  \bibinfo{person}{Flemming~Friche Rodler}.} \bibinfo{year}{2004}\natexlab{}.
\newblock \showarticletitle{Cuckoo hashing}.
\newblock \bibinfo{journal}{\emph{Journal of Algorithms}} \bibinfo{volume}{51},
  \bibinfo{number}{2} (\bibinfo{year}{2004}), \bibinfo{pages}{122 -- 144}.
\newblock
\showISSN{0196-6774}


\bibitem[\protect\citeauthoryear{Pinar et~al\mbox{.}}{Pinar
  et~al\mbox{.}}{2005}]%
        {bitmap2}
\bibfield{author}{\bibinfo{person}{Ali Pinar} {et~al\mbox{.}}}
  \bibinfo{year}{2005}\natexlab{}.
\newblock \showarticletitle{{Compressing Bitmap Indices by Data
  Reorganization}}. In \bibinfo{booktitle}{\emph{{ICDE}}}.
  \bibinfo{pages}{310--321}.
\newblock


\bibitem[\protect\citeauthoryear{Rao and Ross}{Rao and Ross}{1998}]%
        {rao1998csstree}
\bibfield{author}{\bibinfo{person}{Jun Rao} {and} \bibinfo{person}{Kenneth~A
  Ross}.} \bibinfo{year}{1998}\natexlab{}.
\newblock \showarticletitle{Cache conscious indexing for decision-support in
  main memory}. In \bibinfo{booktitle}{\emph{VLDB}}. \bibinfo{pages}{78--89}.
\newblock


\bibitem[\protect\citeauthoryear{Rao and Ross}{Rao and Ross}{2000}]%
        {csb}
\bibfield{author}{\bibinfo{person}{Jun Rao} {and} \bibinfo{person}{Kenneth~A.
  Ross}.} \bibinfo{year}{2000}\natexlab{}.
\newblock \showarticletitle{{Making B\({}^{\mbox{+}}\)-Trees Cache Conscious in
  Main Memory}}. In \bibinfo{booktitle}{\emph{{SIGMOD}}}.
  \bibinfo{pages}{475--486}.
\newblock


\bibitem[\protect\citeauthoryear{Seshadri and Swami}{Seshadri and
  Swami}{1995}]%
        {seshadri1995generalized}
\bibfield{author}{\bibinfo{person}{Praveen Seshadri} {and}
  \bibinfo{person}{Arun Swami}.} \bibinfo{year}{1995}\natexlab{}.
\newblock \showarticletitle{Generalized partial indexes}. In
  \bibinfo{booktitle}{\emph{ICDE}}. IEEE, \bibinfo{pages}{420--427}.
\newblock


\bibitem[\protect\citeauthoryear{Shalev-Shwartz and Ben-David}{Shalev-Shwartz
  and Ben-David}{2014}]%
        {shalev2014understanding}
\bibfield{author}{\bibinfo{person}{Shai Shalev-Shwartz} {and}
  \bibinfo{person}{Shai Ben-David}.} \bibinfo{year}{2014}\natexlab{}.
\newblock \bibinfo{booktitle}{\emph{Understanding machine learning: From theory
  to algorithms}}.
\newblock \bibinfo{publisher}{Cambridge university press}.
\newblock


\bibitem[\protect\citeauthoryear{Shamir, Sabato, and Tishby}{Shamir
  et~al\mbox{.}}{2010}]%
        {shamir2010learning}
\bibfield{author}{\bibinfo{person}{Ohad Shamir}, \bibinfo{person}{Sivan
  Sabato}, {and} \bibinfo{person}{Naftali Tishby}.}
  \bibinfo{year}{2010}\natexlab{}.
\newblock \showarticletitle{Learning and generalization with the information
  bottleneck}.
\newblock \bibinfo{journal}{\emph{Theoretical Computer Science}}
  \bibinfo{volume}{411}, \bibinfo{number}{29-30} (\bibinfo{year}{2010}),
  \bibinfo{pages}{2696--2711}.
\newblock


\bibitem[\protect\citeauthoryear{Stonebraker}{Stonebraker}{1989}]%
        {stonebraker1989case}
\bibfield{author}{\bibinfo{person}{Michael Stonebraker}.}
  \bibinfo{year}{1989}\natexlab{}.
\newblock \showarticletitle{The case for partial indexes}.
\newblock \bibinfo{journal}{\emph{SIGMOD Record}} \bibinfo{volume}{18},
  \bibinfo{number}{4} (\bibinfo{year}{1989}), \bibinfo{pages}{4--11}.
\newblock


\bibitem[\protect\citeauthoryear{Toss, Pahins, Raffin, and Comba}{Toss
  et~al\mbox{.}}{2018}]%
        {toss2018packed}
\bibfield{author}{\bibinfo{person}{Julio Toss}, \bibinfo{person}{Cicero~AL
  Pahins}, \bibinfo{person}{Bruno Raffin}, {and} \bibinfo{person}{Jo{\~a}o~LD
  Comba}.} \bibinfo{year}{2018}\natexlab{}.
\newblock \showarticletitle{Packed-Memory Quadtree: A cache-oblivious data
  structure for visual exploration of streaming spatiotemporal big data}.
\newblock \bibinfo{journal}{\emph{Computers \& Graphics}}  \bibinfo{volume}{76}
  (\bibinfo{year}{2018}), \bibinfo{pages}{117--128}.
\newblock


\bibitem[\protect\citeauthoryear{Wang, Zhang, Sebe, Shen, et~al\mbox{.}}{Wang
  et~al\mbox{.}}{2017}]%
        {wang2017survey}
\bibfield{author}{\bibinfo{person}{Jingdong Wang}, \bibinfo{person}{Ting
  Zhang}, \bibinfo{person}{Nicu Sebe}, \bibinfo{person}{Heng~Tao Shen},
  {et~al\mbox{.}}} \bibinfo{year}{2017}\natexlab{}.
\newblock \showarticletitle{A survey on learning to hash}.
\newblock \bibinfo{journal}{\emph{PAMI}} \bibinfo{volume}{40},
  \bibinfo{number}{4} (\bibinfo{year}{2017}), \bibinfo{pages}{769--790}.
\newblock


\bibitem[\protect\citeauthoryear{Wu et~al\mbox{.}}{Wu et~al\mbox{.}}{2006}]%
        {bitmap5}
\bibfield{author}{\bibinfo{person}{Kesheng Wu} {et~al\mbox{.}}}
  \bibinfo{year}{2006}\natexlab{}.
\newblock \showarticletitle{{Optimizing Bitmap Indices with Efficient
  Compression}}.
\newblock \bibinfo{journal}{\emph{{TODS}}} (\bibinfo{year}{2006}),
  \bibinfo{pages}{1--38}.
\newblock


\bibitem[\protect\citeauthoryear{Wu, Pang, Chen, Gao, Zhao, and Xiang}{Wu
  et~al\mbox{.}}{2019}]%
        {wu2019neist}
\bibfield{author}{\bibinfo{person}{Sai Wu}, \bibinfo{person}{Zhifei Pang},
  \bibinfo{person}{Gang Chen}, \bibinfo{person}{Yunjun Gao},
  \bibinfo{person}{Cenjiong Zhao}, {and} \bibinfo{person}{Shili Xiang}.}
  \bibinfo{year}{2019}\natexlab{}.
\newblock \showarticletitle{NEIST: a Neural-Enhanced Index for Spatio-Temporal
  Queries}.
\newblock \bibinfo{journal}{\emph{TKDE}} (\bibinfo{year}{2019}).
\newblock


\bibitem[\protect\citeauthoryear{Yang, Liang, Kamsetty, Wu, Duan, Chen, Abbeel,
  Hellerstein, Krishnan, and Stoica}{Yang et~al\mbox{.}}{2019}]%
        {yang2019selectivity}
\bibfield{author}{\bibinfo{person}{Zongheng Yang}, \bibinfo{person}{Eric
  Liang}, \bibinfo{person}{Amog Kamsetty}, \bibinfo{person}{Chenggang Wu},
  \bibinfo{person}{Yan Duan}, \bibinfo{person}{Xi Chen},
  \bibinfo{person}{Pieter Abbeel}, \bibinfo{person}{Joseph~M Hellerstein},
  \bibinfo{person}{Sanjay Krishnan}, {and} \bibinfo{person}{Ion Stoica}.}
  \bibinfo{year}{2019}\natexlab{}.
\newblock \showarticletitle{Selectivity Estimation with Deep Likelihood
  Models}.
\newblock \bibinfo{journal}{\emph{arXiv preprint arXiv:1905.04278}}
  (\bibinfo{year}{2019}).
\newblock


\end{thebibliography}

\end{document}